\documentclass[twocolumn]{autart}    

\usepackage{graphicx}          
\usepackage{color}

\usepackage{caption}

\usepackage{hyperref}
\hypersetup{hidelinks} 

\usepackage{natbib}        
\usepackage{subfigure}
\usepackage{float}
\usepackage{algorithm}
\usepackage{amssymb}     
\usepackage{amsmath}
\usepackage{bm}
\usepackage{booktabs} 
\usepackage{etoolbox}
\usepackage{array}

\usepackage{graphicx}      

\usepackage{subfigure}
\usepackage{float}
\usepackage{amssymb}     
\usepackage{amsmath}
\usepackage{bm}
\usepackage{booktabs} 

\usepackage{color}
\usepackage{multirow}
\usepackage{makecell}

\let\classAND\AND
\let\AND\relax
\usepackage{algorithmic}

\let\AND\classAND
\AtBeginEnvironment{algorithmic}{\let\AND\algoAND}

\newenvironment{proof}{\textbf{Proof:}}{\hfill  $\blacksquare$}

\usepackage{array}
\renewcommand\arraystretch{0.75}

\begin{document}

\begin{frontmatter}

\title{Asynchronous Approximate Byzantine Consensus: \\ A Multi-hop Relay Method and Tight Graph Conditions\thanksref{footnoteinfo}} 

\thanks[footnoteinfo]{This work was supported in the part by JSPS under Grants-in-Aid for 
	Scientific Research Grant No.~18H01460 and 22H01508. The support provided by the China
	Scholarship Council is also acknowledged.
    The material in this paper was partially presented at the American Control Conference (ACC2022), June 8--10, 2022, Atlanta, GA, USA.}

\author[Changsha]{Liwei Yuan}\ead{yuanliwei@hnu.edu.cn},    
\author[Tokyo]{Hideaki Ishii}\ead{ishii@c.titech.ac.jp}               

\address[Changsha]{College of Electrical and Information Engineering, Hunan University, Changsha 410082, China}
\address[Tokyo]{Department of Computer Science, Tokyo Institute of Technology, Yokohama 226-8502, Japan}


\begin{abstract}                          
We study a multi-agent resilient consensus problem, where some agents are of the Byzantine type and try to prevent the normal ones from reaching consensus.
In our setting, normal agents communicate with each other asynchronously over multi-hop relay channels with delays. 
To solve this asynchronous Byzantine consensus problem,
we develop the multi-hop weighted mean subsequence reduced (MW-MSR) algorithm. 
The main contribution is that we characterize a tight graph condition for our algorithm to achieve Byzantine consensus, which is expressed in the novel notion of strictly robust graphs. We show that the multi-hop communication is effective for enhancing the network's resilience against Byzantine agents. 
As a result, we also obtain novel conditions for resilient consensus under the malicious attack model, which are tighter than those known in the literature.
Furthermore, the proposed algorithm can be viewed as a generalization of the conventional flooding-based algorithms, with less computational complexity.
Lastly, we provide numerical examples to show the effectiveness of the proposed algorithm.
\end{abstract}

\end{frontmatter}

\section{Introduction}

\begin{table*}[t]\scriptsize
	\caption{Graph conditions for resilient consensus under different adversary models and update schemes.}
	\label{table1}
	\renewcommand\arraystretch{2.0}
	\begin{tabular}{c|c|c|c} 
		\hline 
		\multicolumn{2}{c|}{}  &  
		Synchronous & 
		Asynchronous\\
		\hline
		\raisebox{-4mm}{\multirow{2}*{Malicious}} & $f$-total &
		\makecell[c]{$(f+1,f+1)$-robust  with $l$ hops$ ^\bigstar$  \\
			(\cite{leblanc2013resilient} (for $l=1$);\\   
			\cite{yuan2021resilient})} & 
		\makecell[c]{$(2f+1)$-robust with $l$ hops$^\vartriangle$ \\ 
			(\cite{dibaji2017resilient} (for $l=1$);\\ \cite{yuan2021resilient}) \\ 
			A tighter condition in Corollary~\ref{asynlocalmalicious}: \\
			$(f+1)$-strictly robust with $l$ hops$^\vartriangle$ }\\[2mm]
		\cline{2-4}
		& $f$-local & 
		\makecell[c]{$(2f+1)$-robust with $l$ hops$^\vartriangle$  \\ 
			(\cite{leblanc2013resilient} (for $l=1$);\\
			\cite{yuan2021resilient}) \\ 
			A tighter condition in Corollary~\ref{synlocalmalicious}: \\
			$(f+1)$-strictly robust with $l$ hops$^\vartriangle$} & 
		\makecell[c]{$(2f+1)$-robust with $l$ hops$^\vartriangle$ \\ 
			(\cite{dibaji2017resilient} (for $l=1$);\\ 
			\cite{yuan2021resilient}) \\ 
			A tighter condition in Corollary~\ref{asynlocalmalicious}: \\
			$(f+1)$-strictly robust with $l$ hops$^\vartriangle$ }\\[2mm]
		\cline{1-4}
		\raisebox{-2mm}{\multirow{2}*{Byzantine}} & 
		$f$-total &
		\makecell[c]{$(f+1)$-strictly robust with $l$ hops$^\bigstar$  \\ 
			(\cite{vaidya2012iterative} (for $l=1$);\\  
			\cite{su2017reaching}; This work: Proposition~\ref{syn})} & 
		\makecell[c]{$(f+1)$-strictly robust with $l$ hops$^\bigstar$ \\  
			(\cite{sakavalas2020asynchronous}; 
			This work: Theorem~\ref{asyntheorem}) }\\[2mm]
		\cline{2-4}
		& $f$-local & 
		\makecell[c]{$(f+1)$-strictly robust with $l$ hops$^\bigstar$ \\ 
			(This work: Proposition~\ref{syn})} & 
		\makecell[c]{$(f+1)$-strictly robust with $l$ hops$^\bigstar$ \\ 
			(This work: Theorem~\ref{asyntheorem})}\\[2mm]
		\hline
	\end{tabular}\\[1mm]
	\mbox{}Note that the notion of (strict) robustness is different under the $f$-total and $f$-local models (see Section \ref{secrobustness}). Here, $^\vartriangle$ means sufficient, while $^\bigstar$ means necessary and sufficient.
\end{table*}

With concerns for cyber-security sharply rising in multi-agent systems, consensus resilient in the presence of adversarial agents has gained much attention (\cite{vaidya2012iterative, leblanc2013resilient, dibaji2017resilient,mitra2019byzantine,tian2019chance,yuan2021secure}). 
The focus of this paper is resilient consensus under the Byzantine agents that behave arbitrarily, which has a rich history in distributed
computing (\cite{dolev1982byzantine, Lynch}).
\cite{dolev1986reaching} introduced the above so-called \textit{approximate Byzantine consensus} problem for the case of complete networks, where the non-adversarial nodes are required to achieve \textit{approximate agreement} by converging to a relatively small interval in finite time. \cite{vaidya2012iterative}, \cite{su2017reaching} studied the same problem for synchronous networks with general topologies; see also \cite{leblanc2013resilient}.
To solve the asynchronous version of the problem, flooding-based algorithms were proposed in \cite{abraham2004optimal,sakavalas2020asynchronous}.

In this paper, we study the asynchronous approximate Byzantine consensus problem using mean subsequence reduced (MSR) algorithms, which are often used for iterative fault-tolerant consensus algorithms (\cite{azadmanesh2002asynchronous,  vaidya2012iterative, bonomi2019approximate}).
In MSR algorithms, normal nodes discard the most deviated states from neighbors to avoid being influenced by possible extreme values from adversaries. Moreover, graph robustness is found to be a tight graph condition for the network using MSR algorithms (\cite{leblanc2013resilient,abbas2017improving, dibaji2018resilient,  wang2019resilient,lu2023distributed}).

In this context, we focus on using multi-hop relay techniques to relax the heavy graph connectivity requirement for Byzantine consensus. 
Multi-hop communication enables networks to have multiple paths for interactions among nodes (\cite{Lynch,goldsmith2005wireless}), and hence, it is effective for enhancing resilience against node failures.
In the systems and control area, there are works analyzing the stability of the networked control systems with control inputs and observer information sent over multi-hop networks (\cite{dinnocenzo2016resilient,cetinkaya2018probabilistic}).
Multi-hop techniques are also used for consensus problems in recent years (\cite{jin2006multi,zhao2016global,ding2021analysis}). 
Recently, \cite{su2017reaching} introduced multi-hop communication in MSR algorithms, and they solved the synchronous Byzantine consensus problem with a weaker condition on network structures compared to that derived under the one-hop case in \cite{vaidya2012iterative}. 
Later, \cite{sakavalas2020asynchronous} studied the asynchronous Byzantine consensus problem using a flooding-based algorithm.
However, strong assumptions were made in \cite{sakavalas2020asynchronous}. Specifically, their algorithm essentially requires each normal node to be aware of the global topology information and to ``flood'' its own value until it is relayed to reach all nodes in the network.
In \cite{yuan2021resilient}, we extended the notion of (one-hop) graph robustness to the multi-hop case and provided a tight necessary and sufficient graph condition for resilient consensus under the malicious model. Table~\ref{table1} summarizes related resilient consensus works.

The contributions of this paper are outlined as follows.
First, we study the asynchronous Byzantine consensus problem using the multi-hop weighted MSR (MW-MSR) algorithm proposed in our previous work (\cite{yuan2021resilient}) for the malicious model. As a main contribution, we prove a tight necessary and sufficient condition for our algorithm to achieve asynchronous Byzantine consensus, expressed by the novel notion of \textit{strictly robust graphs with $l$ hops.}
This condition requires more connections than the robustness notion in \cite{yuan2021resilient} since the Byzantine model is more adversarial than the malicious model.
Compared to \cite{sakavalas2020asynchronous}, our algorithm is more light-weighted and makes a weaker assumption as mentioned earlier. Moreover, the problem there is a special case of multi-hop paths with unbounded lengths in this paper. Their graph condition also coincides with ours by setting the path length to be the longest cycle-free one in the graph. Since in our model, the number of hops is limited, our approach is more distributed in the sense that we only require each normal node to know the local topology and neighbors' values up to $l$ hops away.

Most importantly, our algorithm can achieve Byzantine consensus under the $f$-local model, which is even more adversarial than the $f$-total model in \cite{su2017reaching,sakavalas2020asynchronous,yuan2021resilient}. As we show in Section 6, our algorithm can tolerate more Byzantine agents in the network than the above works for the $f$-total model.
Such an advantage is because of the flexibility of our algorithm on general $l$-hop communication.
In contrast, the flooding algorithm in \cite{sakavalas2020asynchronous} is restricted to the $f$-total model since a Byzantine node there can flood erroneous values to all the nodes in the network.
From a practical point of view, our approach offers an adjustable option for the trade-off between an appropriate level of resilience and an affordable cost of communication resources. Even in the same network with Byzantine agents,
our method generally requires less relay hops to achieve Byzantine consensus compared to \cite{sakavalas2020asynchronous}.

Lastly, this paper also provides novel insights into resilient consensus under the malicious model. 
It turns out that the robust graph conditions known in the literature can be tightened even for the one-hop synchronous $f$-local model and asynchronous $f$-local/total model. 
Moreover, we obtain a tighter sufficient graph condition for the multi-hop case in \cite{yuan2021resilient}, and we also extend the results to the $f$-local model.
These results are indicated in Table~\ref{table1}.
The key contribution for these advances is that we prove the order between the different graph conditions for the two adversary models (Proposition~\ref{threeconditions}) for general $l$-hop communication. 
Moreover, our approach can be easily extended to more complex multi-agent consensus systems, e.g., agents with second order dynamics (\cite{dibaji2017resilient}) and resilient consensus-based formation control problems.

Compared to the conference paper (\cite{yuan2022asynchronous}), 
the current paper contains the following novel contents: the analysis of the synchronous MW-MSR algorithm on the $f$-local Byzantine model, the relations between different graph conditions, a tighter condition for the resilient consensus under the malicious model, as well as novel simulations.
We also present an event-triggered scheme to our algorithm for the Byzantine consensus problem in \cite{yuan2022event}.

The rest of this paper is organized as follows. 
Section~2 outlines preliminaries and the system model. Section~3 presents the notion and properties of strictly robust graphs with $l$ hops.
In Sections 4 and 5, we derive conditions under which the MW-MSR algorithms guarantee Byzantine consensus under synchronous and asynchronous updates, respectively.
Section~6 provides numerical examples to demonstrate the efficacy of our algorithm.
Lastly, Section~7 concludes the paper.

\section{Preliminaries and Problem Setting}

In this section, we introduce the problem setting of this paper and outline our resilient consensus algorithm.

\subsection{Network Model under Multi-hop Communication}
Consider the directed graph $\mathcal{G} = (\mathcal{V},\mathcal{E})$ consisting of the node set $\mathcal{V}=\{1,...,n\}$ and the edge set $\mathcal{E}\subset \mathcal{V} \times \mathcal{V}$. Here, the edge $(j,i)\in \mathcal{E}$ indicates that node $i$ can get information from node $j$. The subgraph of $\mathcal{G} = (\mathcal{V},\mathcal{E})$ induced by the node set $\mathcal{H}\subset\mathcal{V}$ is the subgraph $\mathcal{G}_\mathcal{H}=(\mathcal{V}(\mathcal{H}),\mathcal{E}(\mathcal{H}))$, where $\mathcal{V}(\mathcal{H})=\mathcal{H}$, $\mathcal{E}(\mathcal{H})=\{(i,j)\in \mathcal{E}: i,j\in \mathcal{H}\}$.
An $l$-hop path from node $i_1$ to $i_{l+1}$ is a sequence of distinct nodes $(i_1, i_2, \dots, i_{l+1})$, where $(i_j, i_{j+1})\in \mathcal{E} $ for $j=1, \dots, l$. Node $i_{l+1}$ is reachable from node $i_1$. 
Let $\mathcal{N}_i^{l-}$ be the set of nodes that can reach node $i$ via at most $l$-hop paths.
Let $\mathcal{N}_i^{l+}$ be the set of nodes that are reachable from node $i$ via at most $l$-hop paths. 
The $l$-th power of the graph $\mathcal{G}$, denoted by $\mathcal{G}^l$, is a multigraph with $\mathcal{V}$ and a directed edge from node $j$ to node $i$ is defined by a path of length at most $l$ from $j$ to $i$ in $\mathcal{G}$. The adjacency matrix $A = [a_{ij} ]$ of $\mathcal{G}^l$ is given by $\alpha \leq a_{ij}<1$ if $j\in \mathcal{N}_i^{l-}$ and otherwise $a_{ij} = 0$, where $\alpha > 0$ is fixed and $\sum_{j=1,j\neq i}^{n} a_{ij}\leq 1$. Let $L = [b_{ij} ]$ be the Laplacian matrix of $\mathcal{G}^l$, where $b_{ii} =\sum_{j=1,j\neq i}^{n}a_{ij}$, $b_{ij} = -a_{ij}$ for $ i\neq j$.

Next, we describe our communication model, inspired by \cite{su2017reaching,yuan2021resilient}. Node $i_1$ can send messages of its own to an $l$-hop neighbor $i_{l+1}$ via different paths.
We represent a message as a tuple $m=(w,P)$, where $w=\mathrm{value}(m)\in \mathbb{R}$ is the message content and $P=\mathrm{path}(m)$ indicates the path via which message $m$ is transmitted. 
At time $k\geq 0$, each normal node $i$ exchanges the messages $m_{ij}[k]=(x_i[k],P_{ij}[k])$ consisting of its state $x_i[k]$ along each path $P_{ij}[k]$ with its multi-hop neighbor $j$ via the relaying process in \cite{yuan2021resilient}.
Denote by $\mathcal{V}(P)$ the set of nodes in $P$.

\subsection{Update Rule and Threat Models}\label{problemsetting}
Consider a time-invariant directed network $\mathcal{G} = (\mathcal{V},\mathcal{E})$.
The node set $\mathcal{V}$ is partitioned into the set of normal nodes $\mathcal{N}$ and the set of adversary nodes $\mathcal{A}$, where $n_N=|\mathcal{N}|$ and $n_A=|\mathcal{A}|$. The partition is unknown to the
normal nodes at all times.


When there is no attack, we use the consensus update rule extended from \cite{olfati2007consensus}, given as
\begin{equation}\label{m1}
	\begin{aligned}
		x[k+1]=x[k] +  u[k], \ \ \ 
		u[k]=-L[k]x[k],
	\end{aligned}
\end{equation}
where $x[k]\in \mathbb{R}^n$ and $u[k]\in \mathbb{R}^n$ are the state vector and the control input vector, respectively.
The power graph $\mathcal{G}^l[k]$ at time $k$ is determined by the messages used for updates by the agents, i.e., $m_{ji}[k]$ for $ i\in \mathcal{V}$ and $j\in \mathcal{N}_i^{l-}$. The adjacency matrix $A[k]$ and the Laplacian matrix $L[k]$ at time $k$ are determined accordingly. Considering possible adversaries in $\mathcal{A}$, normal nodes use the resilient algorithm to be presented later for updating their values.



We introduce the threat models extended from those studied in \cite{vaidya2012iterative}, \cite{leblanc2013resilient}.
\begin{defn}
	\textit{($f$-total/$f$-local set)}
	The set of adversary nodes $\mathcal{A}$ is said to be $f$-total
	if it contains at most $f$ nodes, i.e., $\left| \mathcal{A}\right| \leq f$.
	Similarly, it is said to be $f$-local (in $l$-hop neighbors)
	if any normal node $i$ has at most $f$ adversary nodes as its $l$-hop neighbors, i.e., $\left|\mathcal{N}_i^{l-} \cap \mathcal{A}\right| \leq f$.
\end{defn}

\begin{defn}
	\textit{(Byzantine nodes)}
	An adversary node $i\in \mathcal{A}$ is said to be Byzantine
	if it can arbitrarily modify its own value and relayed values and sends different state values and relayed values to its neighbors at each step.
\end{defn}

The Byzantine model is well studied in computer science (\cite{dolev1982byzantine, Lynch, vaidya2012iterative}). Note that the \textit{malicious} model studied in \cite{leblanc2013resilient}, \cite{dibaji2017resilient} is a weaker threat model as malicious nodes must send the same information to their neighbors, which is suitable for broadcast networks. We should also note that for the multi-hop communication case, the malicious model is considered in \cite{yuan2021resilient}.

As commonly done in the literature, 
we assume that each normal node knows the value of $f$ and the topology information of the graph up to $l$ hops. 
Moreover, to keep the problem tractable, we introduce the following assumption (\cite{su2017reaching}).
It is merely introduced for ease of analysis. In fact, manipulating message paths can be easily detected and hence does not create problems. 
We have shown how this can be done in \cite{yuan2021resilient}, inspired by \cite{su2017reaching}.

\begin{assum}
	Each Byzantine node $i$ can manipulate its state $x_i[k]$ and the values in messages that they send or relay, but cannot change the path $P$ in such messages. 
\end{assum}

\subsection{Resilient Asymptotic Consensus and Algorithm}

We define the resilient consensus notion used in this paper, which is also studied in, e.g., \cite{leblanc2013resilient}, \cite{su2017reaching},  \cite{dibaji2017resilient}.

\begin{defn}
	If for any possible sets and behaviors of the
	adversaries and any state values of the normal
	agents, the following conditions are satisfied,
	then we say that the normal agents reach 
	resilient asymptotic consensus:
	
	\begin{enumerate}
		\item Safety: There exists a bounded safety interval $\mathcal{S}$ determined by the initial values of the normal agents such that $x_i[k] \in \mathcal{S}, \forall i \in \mathcal{N}, k \in \mathbb{Z}_+$. 
		\item Agreement: There exists a state $x^*\in \mathcal{S}$
		such that $\lim_{k\to \infty}x_i[k]=x^*,  \forall i\in \mathcal{N}$.
	\end{enumerate}
	
\end{defn}

Next, we present the multi-hop weighted-MSR (MW-MSR) algorithm from our previous work (\cite{yuan2021resilient}) in Algorithm~1. The notion of minimum message cover (MMC) (\cite{su2017reaching}) is crucial in Algorithm~1, which is defined as follows.

\begin{defn} For a graph $\mathcal{G} = (\mathcal{V},\mathcal{E})$, let $\mathcal{M}$ be a set of messages transmitted through $\mathcal{G}$, and let $\mathcal{P}(\mathcal{M})$ be the set of message paths of all the messages in $\mathcal{M}$, i.e., $\mathcal{P}(\mathcal{M}) =\{\mathrm{path}(m):m \in \mathcal{M}\}$. A \textit{message cover} of $\mathcal{M}$ is a set of nodes $\mathcal{T}(\mathcal{M})\subset \mathcal{V}$ whose removal disconnects all message paths, i.e., for each path $P\in \mathcal{P}(\mathcal{M})$, we have $\mathcal{V}(P)\cap \mathcal{T}(\mathcal{M})\neq Ø$. In particular, a \textit{minimum message cover} of $\mathcal{M}$ is defined by
	\begin{equation*}
		\mathcal{T}^*(\mathcal{M})\in	\arg \min_{\substack{ \mathcal{T}(\mathcal{M}): \textup{ Cover of } \mathcal{M}}} 	\left|  \mathcal{T} (\mathcal{M})\right| . 
	\end{equation*}
\end{defn}
\vspace*{-5.5mm}

\begin{algorithm}[t] 
	\caption{: MW-MSR Algorithm} 
	\begin{algorithmic}
		\STATE 1) At each time $k$, for $\forall i \in \mathcal{N}$:
		
		\STATE Send $m_{ij}[k]=(x_i[k],P_{ij}[k])$ to $\forall j\in \mathcal{N}_i^{l+}$. 
		
		\STATE Receive $m_{ji}[k]=(x_j[k],P_{ji}[k])$ from $\forall j\in \mathcal{N}_i^{l-}$ and store them in the set $\mathcal{M}_i[k]$.
		
		\STATE Sort $\mathcal{M}_i[k]$ in an increasing order based on the message values (i.e., $x_j[k]$ in $m_{ji}[k]$).
		
		\STATE 2) Remove extreme values:

		\STATE (a) Define two subsets of $\mathcal{M}_i[k]$:
		\vspace{1mm}
		
		\STATE \hspace{2.2mm} $\overline{\mathcal{M}}_i[k]=\{ m\in \mathcal{M}_i[k]: \mathrm{value}(m)> x_i[k]  \}$,
		\vspace{1mm}
		
		\STATE \hspace{2.2mm} $\underline{\mathcal{M}}_i[k]=\{ m\in \mathcal{M}_i[k]: \mathrm{value}(m)< x_i[k]  \}$.
		\vspace{1mm}
		
		\STATE (b) Get $\overline{\mathcal{R}}_i[k]$ from $\overline{\mathcal{M}}_i[k]$:

		\IF{$\left|  \mathcal{T}^* (\overline{\mathcal{M}}_i[k])\right| <f$}
		\STATE $\overline{\mathcal{R}}_i[k] = \overline{\mathcal{M}}_i[k]$;
		\ELSE
		\STATE Choose $\overline{\mathcal{R}}_i[k]$ s.t. (i)
		$\forall m\in \overline{\mathcal{M}}_i[k]\setminus \overline{\mathcal{R}}_i[k]$, $\forall m'\in \overline{\mathcal{R}}_i[k]$, $\mathrm{value}(m) \leq \mathrm{value}(m')  \medspace\medspace \medspace \medspace  $ 
		and (ii) $\left|  \mathcal{T}^* (\overline{\mathcal{R}}_i[k])\right| =f$. 
		\ENDIF
		
		\STATE (c) Get $\underline{\mathcal{R}}_i[k]$ from $\underline{\mathcal{M}}_i[k]$ similarly, which contains smallest message values.
		
		\STATE (d) $\mathcal{R}_i[k]=\overline{\mathcal{R}}_i[k]\cup\underline{\mathcal{R}}_i[k]$.

		\STATE 3) Update:
		$a_{i}[k]=1/(\left| \mathcal{M}_i[k]\setminus \mathcal{R}_i[k] \right| )$,
		\begin{equation}
			x_i[k+1]=\sum_{m\in \mathcal{M}_i[k]\setminus \mathcal{R}_i[k]} a_{i}[k] \medspace \mathrm{value}(m).  \label{updaterule2}
		\end{equation}
	\end{algorithmic}
\end{algorithm}

In Algorithm~1, normal node $i$ can trim away the largest and smallest values from exactly $f$ nodes within $l$ hops away.
Clearly, as the number $l$ of hops grows, the candidate nodes increase and the trimming step 2 becomes more complicated. The reason is that in the multi-hop setting, each node relays the values from different neighbors, node $i$ can receive more than one value from one direct neighbor at each step.
For more details of Algorithm~1, we refer to \cite{yuan2021resilient}.

To characterize the number of the extreme values from exactly $f$ nodes for node $i$, the notion of minimum message cover (MMC) is designed. 
Intuitively speaking, for normal node $i$, $\overline{\mathcal{R}}_i[k]$ and $\underline{\mathcal{R}}_i[k]$ are the largest sized sets of received messages containing very large and small values that may have been generated or tampered by $f$ adversary nodes, respectively. 
Here, we focus on how $\overline{\mathcal{R}}_i[k]$ is determined (depicted in Fig.~\ref{flow}), as $\underline{\mathcal{R}}_i[k]$ can be obtained in a similar way. 
When the cardinality of the MMC of set $\overline{\mathcal{M}}_i[k]$ (in step 2(a)) is no more than $f$, node $i$ simply takes $\overline{\mathcal{R}}_i[k]=\overline{\mathcal{M}}_i[k]$. 
Otherwise, node $i$ will check the largest $q:=f+1$ values of $\overline{\mathcal{M}}_i[k]$, and if the MMC of these values is of cardinality $f$, then it will check the first $q=q+1$ values of $\overline{\mathcal{M}}_i[k]$. This procedure will continue until for the first $q$ values of $\overline{\mathcal{M}}_i[k]$, the MMC of these values is of cardinality $f+1$. Then $\overline{\mathcal{R}}_i[k]$ is taken as the first $q-1$ values of $\overline{\mathcal{M}}_i[k]$.
After sets $\overline{\mathcal{R}}_i[k]$ and $\underline{\mathcal{R}}_i[k]$ are determined, in step~3, node $i$ excludes the values in these sets and updates its value using the remaining values in $\mathcal{M}_i[k]\setminus \mathcal{R}_i[k]$.

\begin{figure}[t]
	\centering
	\includegraphics[width=2.9in]{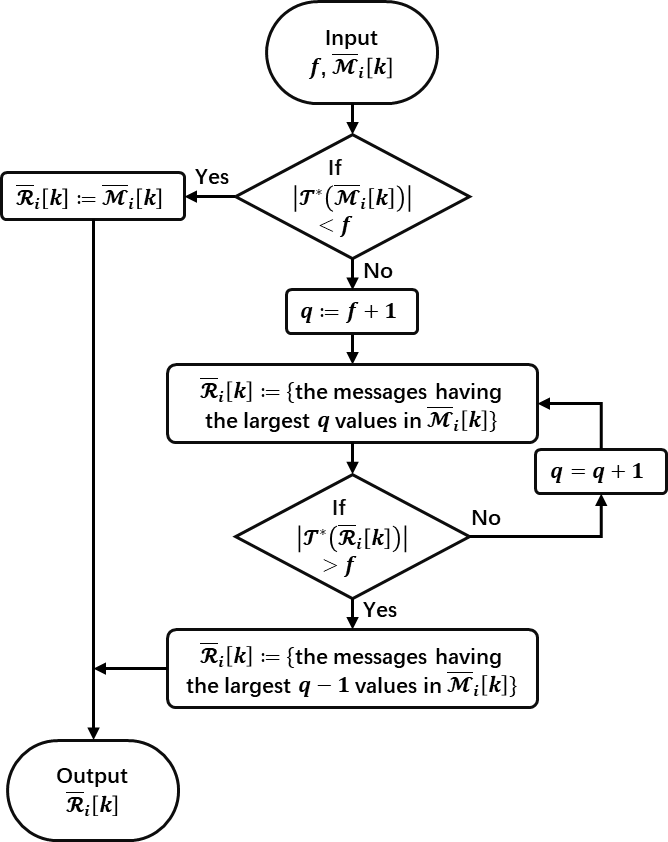}
	\caption{The flow of determining $\overline{\mathcal{R}}_i[k]$.}
	\label{flow}
\end{figure}

In this paper, the main goal is to characterize the conditions on the network structure that guarantee approximate asynchronous Byzantine consensus using the MW-MSR algorithm. Before proceeding to such an analysis, in Section 3, we introduce the important notion of strictly robust graphs. In Section~4, we first consider the case of synchronous updates. Then, in Section~5, we consider the more realistic situation using multi-hop techniques, which is the asynchronous updates with time delays in the communication among agents.

\section{Strictly Robust Graphs with Multi-hop Communication}\label{secrobustness}

In this section, we provide the definition of \textit{strictly robust graphs with $l$ hops}, which is the key graph condition to guarantee Byzantine consensus.

\subsection{The Notion of $(r,s)$-Robust Graphs with $l$ Hops}

The notion of robust graphs was first introduced in \cite{leblanc2013resilient}, and it was proved that graph robustness gives a tight graph condition for MSR-based algorithms guaranteeing resilient consensus under the malicious model. In \cite{yuan2021resilient}, we generalized this notion to the multi-hop case. Its definition is given as follows.

\begin{defn}\label{rs-robust} A directed graph $\mathcal{G} = (\mathcal{V},\mathcal{E})$ is said to be $(r,s)$-robust with $l$ hops with respect to a given set $\mathcal{F}\subset \mathcal{V}$,
	if for every pair of nonempty disjoint subsets $\mathcal{V}_\text{1},\mathcal{V}_\text{2}\subset \mathcal{V}$, at least one of the following conditions holds:
	
	(1) $\mathcal{Z}_{\mathcal{V}_1}^r=\mathcal{V}_1$; 
	(2) $\mathcal{Z}_{\mathcal{V}_2}^r=\mathcal{V}_2$;
	(3) $\left| \mathcal{Z}_{\mathcal{V}_1}^r\right| +\left| \mathcal{Z}_{\mathcal{V}_2}^r\right| \geq s$,
	
	
	\noindent where $\mathcal{Z}_{\mathcal{V}_a}^\textit{r}$ is the set of nodes in $\mathcal{V}_\textit{a}$ ($a=1,2$) that have at least $r$ independent paths of at most $l$ hops originating from nodes outside $\mathcal{V}_\textit{a}$ and all these paths do not have any nodes in set $\mathcal{F}$ as intermediate nodes (i.e., the nodes in $\mathcal{F}$ can be source or destination nodes in these paths).
	Moreover, if the graph $\mathcal{G}$ satisfies this property with respect to any set $\mathcal{F}$ satisfying the $f$-local model, then we say that $\mathcal{G}$ is $(r,s)$-robust with $l$ hops under the $f$-local model. When it is clear from the context, we just say $\mathcal{G}$ is $(r,s)$-robust with $l$ hops. If $\mathcal{G}$ is $(r,1)$-robust with $l$ hops, it is also defined as $r$-robust with $l$ hops.
\end{defn}

Intuitively speaking, for any set $\mathcal{F}\subset \mathcal{V}$, and for node $i\in \mathcal{V}_\text{1}$ to have the abovementioned property, there should be at least $r$ source nodes outside $\mathcal{V}_\text{1}$ and at least one independent path of length at most $l$ hops from each of the $r$ source nodes to node $i$, where such a path does not contain any internal nodes from the set $\mathcal{F}$. 
In the multi-hop relay environment, the adversary agents can also manipulate the relayed values. Thus, the robustness with $l$ hops is defined with respect to set $\mathcal{F}$ to characterize the ability of node $i$ receiving the original values of the multi-hop agents.
As an example, node $i\in \mathcal{V}_1$ in Fig.~\ref{3-reachable}(a) has $r=2$ independent paths of at most two hops originating from the nodes outside $\mathcal{V}_1$ with respect to set $\mathcal{F}=\{j\}$, while node $i$ in Fig.~\ref{3-reachable}(b) does not.

\begin{figure}[t]
	\centering
	\subfigure[]{
		\includegraphics[width=1.4in]{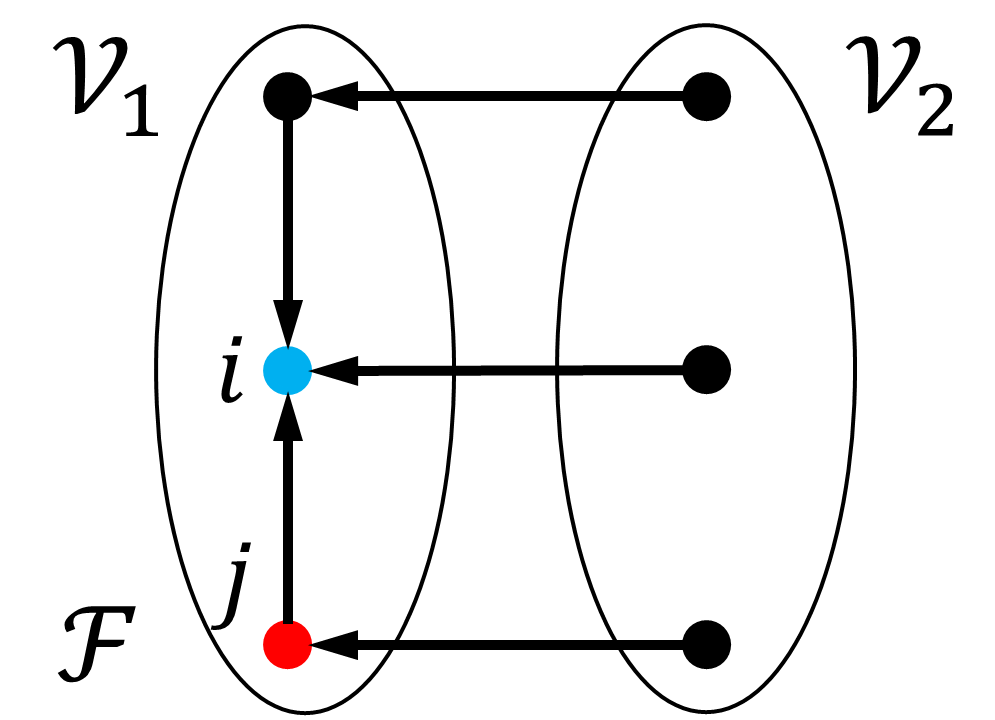}
	}
	\quad
	\subfigure[]{
		\includegraphics[width=1.4in]{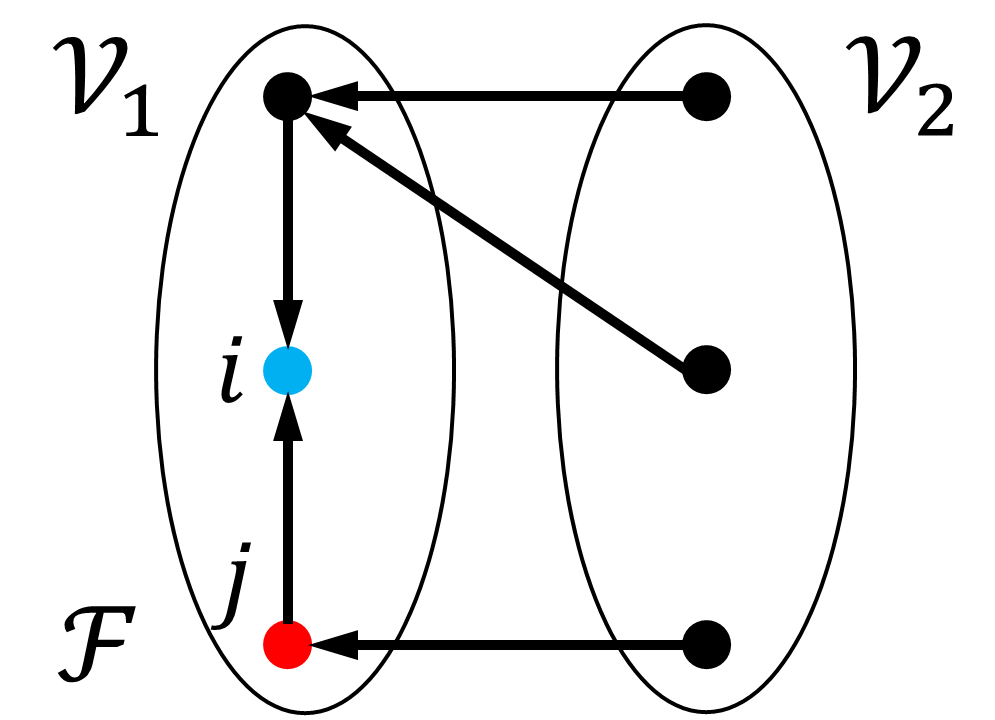}
	}
	\caption{(a) Node $i$ has two independent paths originating from the outside of $\mathcal{V}_1$ and do not go through the nodes in the set $\mathcal{F}=\{j\}$. (b) Node $i$ has only one independent path sharing the same property.}
	\label{3-reachable}
\end{figure}

Here, we provide some properties of robust graphs with $l$ hops (\cite{yuan2021resilient}). Note that all the properties listed coincide with the ones of one-hop case in \cite{leblanc2013resilient} when $l=1$. Here, $\lceil \cdot\rceil$ denotes the ceiling function.

\begin{lem}\label{robustlemma}
	If a graph $\mathcal{G} = (\mathcal{V},\mathcal{E})$ is $(r, s)$-robust with $l$ hops, then the following hold:
	\begin{enumerate}
		\item $\mathcal{G}$ is $(r', s')$-robust with $l$ hops, where $0 \leq r' \leq r, 1 \leq s' \leq s$.
		\item $\mathcal{G}$ is $(r, s)$-robust with $l'$ hops, where $l \leq l'$.
		\item $\mathcal{G}$ is $(r-1, s+1)$-robust with $l$ hops.
		\item  $\mathcal{G}$ has a directed spanning tree. Moreover, if $\mathcal{G}$ is undirected, then it is $r$-connected.
		\item  $r\leq \lceil n/2\rceil$. Moreover, $\mathcal{G}$ is $(r, s)$-robust with $l$ hops if it is $(r+s-1)$-robust with $l$ hops.
	\end{enumerate}

\end{lem}

\subsection{The Notion of $r$-Strictly Robust Graphs with $l$ Hops}\label{sec_strict_robustness}


To deal with the Byzantine model, we need to focus on the subgraph consisting of only the normal nodes.
Define such a subgraph as the \textit{normal network} as follows.

\begin{defn}
	For a network $\mathcal{G} = (\mathcal{V}, \mathcal{E})$, define the normal
	network of $\mathcal{G}$, denoted by $\mathcal{G}_{\mathcal{N}}$, as the network induced by the normal nodes, i.e., $\mathcal{G}_{\mathcal{N}} = (\mathcal{N},\mathcal{E}_{\mathcal{N}} )$, where $\mathcal{E}_{\mathcal{N}}$ is the set of directed edges among the normal nodes.
	
\end{defn}

For the one-hop MSR algorithm in \cite{leblanc2013resilient}, the graph condition that the normal network is $(f+1)$-robust is proved to be necessary and sufficient for achieving resilient consensus under the $f$-total Byzantine model.
However, in practice, the normal nodes are not aware of the identity of the Byzantine nodes. Hence, the above condition can not be checked a priori. Therefore, we define our graph condition on the original graph topology as \cite{vaidya2012iterative}, \cite{su2017reaching} did and we formally introduce $r$-strictly robust graphs with $l$ hops as follows.

\begin{defn}\label{strictrobust}
	Let $\mathcal{F} \subset \mathcal{V}$  and denote the subgraph of $\mathcal{G}$ induced
	by node set $\mathcal{H}=\mathcal{V}\setminus\mathcal{F}$ as $\mathcal{G}_{\mathcal{H}}$.
	Graph $\mathcal{G}$ is said to be $r$-strictly robust with $l$ hops with respect to $\mathcal{F}$ if the subgraph $\mathcal{G}_{\mathcal{H}}$ is $r$-robust with $l$ hops with respect to $\mathcal{F}$ in graph $\mathcal{G}$.\footnote[2]{Note that the removed node set $\mathcal{F}$ is still used to count the robustness of the remaining graph $\mathcal{G}_{\mathcal{H}}$ since strict robustness is a property of the original graph $\mathcal{G}$. Moreover, the current definition brings the connection between the notions of robustness and strict robustness.  }
	If graph $\mathcal{G}$ satisfies this property with respect to any set $\mathcal{F}$ satisfying the $f$-total/local model, then we say that $\mathcal{G}$ is $r$-strictly robust with $l$ hops (under the $f$-total/local model).
\end{defn}

Robustness with $l\geq2$ hops and strict robustness with $l\geq1$ hops depend on the choice of set $\mathcal{F}$. This set further depends on the threat models.
We illustrate how multi-hop relaying can improve strict robustness through examples.
The graphs in Fig.~\ref{1lcoal} are not $2$-strictly robust with $1$ hop, e.g., in Fig.~\ref{1lcoal}(b), if we remove node 3, the remaining graph is not $2$-robust.
The two graphs are however $2$-strictly robust with $2$ hops under the $1$-local model. Note that to verify the strict robustness, we must check that after removing any node set $\mathcal{F}$ being $1$-local, the remaining graph is $2$-robust with $2$ hops. 

\begin{figure}[t]
	\centering
	\subfigure[]{
		\includegraphics[width=1.1in]{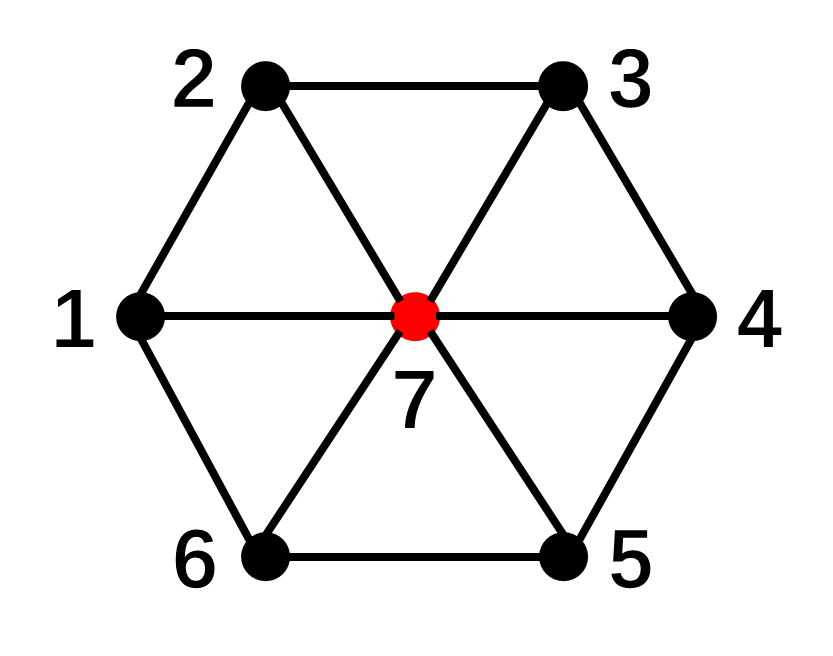}
	}
	\vspace{-3pt}
	\subfigure[]{
		\includegraphics[width=1.8in]{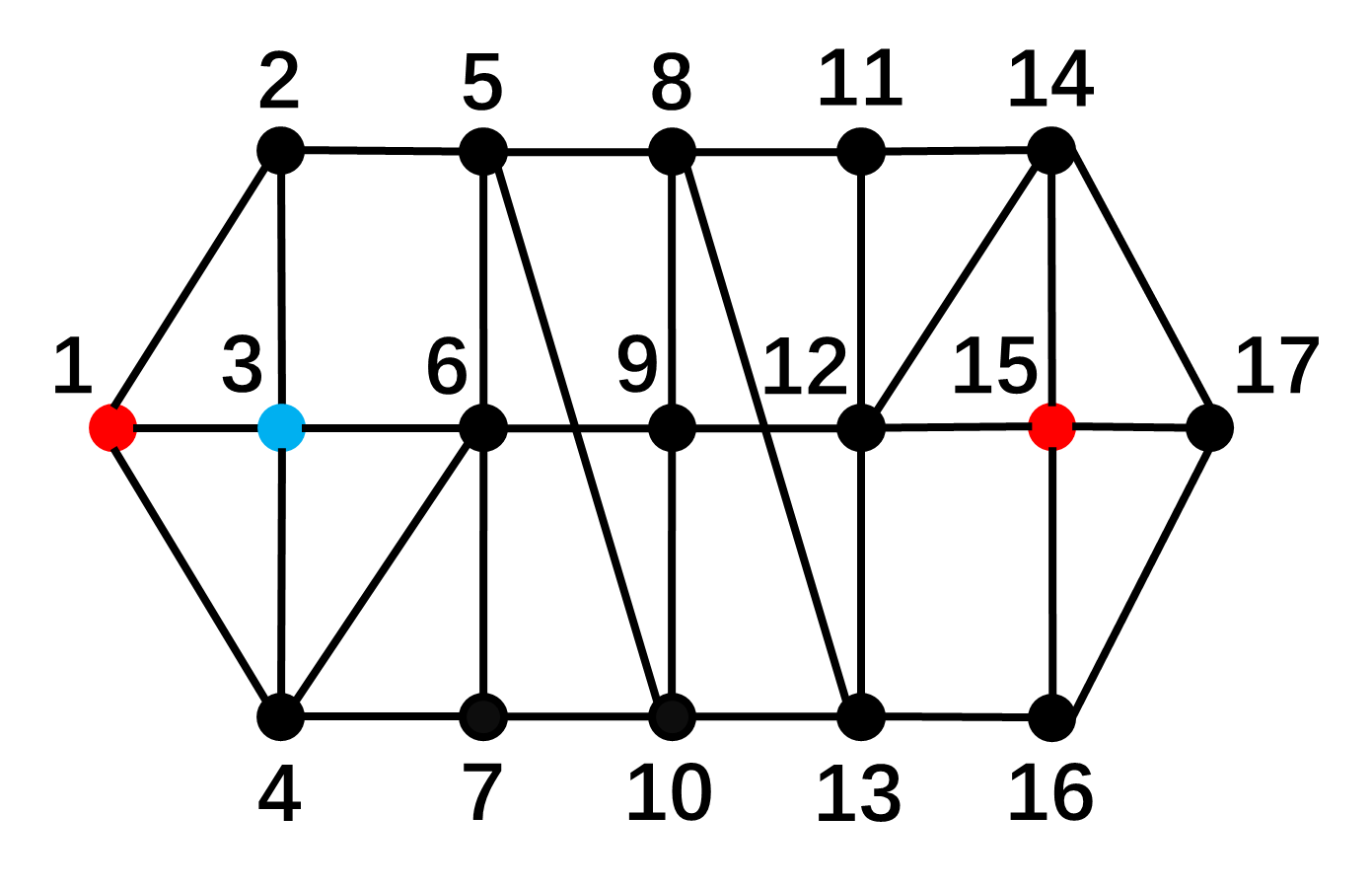}
	}
	\vspace{-3pt}
	\caption{ Both undirected graphs are not $2$-strictly robust with $1$ hop but are $2$-strictly robust with $2$ hops under the 1-local model. }
	\label{1lcoal}
\end{figure}

\section{Synchronous Byzantine Consensus}\label{secsyn}

In this section, we provide the analysis of the MW-MSR algorithm under synchronous updates. 

It is worth noting that \cite{su2017reaching} investigated an MSR-based algorithm with multi-hop communication under the $f$-total Byzantine model. They provided a necessary and sufficient graph condition
for their algorithm to achieve synchronous Byzantine consensus. While their proof techniques are
different, the condition can be interpreted by the notion of strict robustness with $l$ hops as well.
Here, we extend the proof for the $f$-local model, which contains the case of the $f$-total model.
Besides, based on our proof scheme, we can provide the analysis of our algorithm applied in asynchronous updates with delays next in Section~5; such a case is absent in \cite{su2017reaching}.

Denote the vectors consisting of the states of the normal nodes and those of the Byzantine nodes by $x^N[k]$ and $x^A[k]$, respectively. Then, we present the main result of this section in the following.

\begin{prop}\label{syn}
	Consider a directed graph $\mathcal{G} = (\mathcal{V},\mathcal{E})$ with $l$-hop communication, where each normal node updates its value according to the synchronous MW-MSR algorithm with parameter
	$f$. Under the $f$-local Byzantine model, resilient asymptotic
	consensus is achieved with safety interval $\mathcal{S}=\big[ \min x^N[0], \max x^N[0] \big]$ if and only if  $\mathcal{G}$ is $(f + 1)$-strictly robust with $l$ hops.
\end{prop}

\textit{Proof:} (Necessity) If $\mathcal{G}$ is not $(f + 1)$-strictly robust with $l$ hops, then by Definition~\ref{strictrobust}, 
there exists an $f$-local set $\mathcal{F}$ such that $\mathcal{G}$ is not $(f+1)$-strictly robust with $l$ hops with respect to $\mathcal{F}$. Suppose that $\mathcal{F}$ is exactly the set of Byzantine agents, and the normal network $\mathcal{G}_{\mathcal{N}}$ is not $(f +1)$-robust with $l$ hops w.r.t. this $\mathcal{F}$.
In such a case, there are nonempty, disjoint subsets $\mathcal{V}_1, \mathcal{V}_2\subset \mathcal{N}$ such that any node in the two sets has at most $f$ independent paths (only the node itself is common in these paths) of at most $l$ hops originating from normal nodes outside of its respective set.
Let the nodes in the two sets take the maximum and minimum values in the network, respectively. 
Suppose that the Byzantine nodes send the maximum and minimum values to the nodes in $\mathcal{V}_1$ and $ \mathcal{V}_2$, respectively. 


	Consider node $i\in\mathcal{V}_1$. Since the cardinality of the minimum message cover of the values larger than itself (values from the Byzantine nodes) is at most $f$, node $i$ will discard these values. 
	We claim that the cardinality of the minimum message cover of the values smaller than itself (values from the normal nodes outside of $\mathcal{V}_1$) is also at most $f$. This can be proved in three cases: (i) All the incoming neighbors outside of $\mathcal{V}_1$ are direct neighbors of node $i$, (ii) all the incoming neighbors outside of $\mathcal{V}_1$ are $l$-hop ($l\geq 2$) neighbors of node $i$, and (iii) situations other than (i) and (ii). For case (i), it is clear that this statement holds. For case (ii), 
	either node $i$ has at most $f$ independent paths from the $l$-hop ($l\geq 2$) neighbors outside of $\mathcal{V}_1$, where the cardinality of the $l$-hop neighbors can be larger than $f$; or node $i$ has more than $f$ independent paths from the $l$-hop ($l\geq 2$) neighbors outside of $\mathcal{V}_1$, where the cardinality of the $l$-hop neighbors can be at most $f$. In either case, the cardinality of the minimum message cover of the minimum values is at most $f$.
	For case (iii), note that the direct neighbors will be part of the minimum message cover always. For the remaining $l$-hop neighbors outside, following the analysis for case (ii), we can conclude that the cardinality of the minimum message cover of the minimum values is at most $f$. Thus, in all cases, node $i$ discards the values from the nodes outside of $\mathcal{V}_1$ and keeps its value. 

Similar analysis applies when $i\in\mathcal{V}_2$. Therefore, nodes in these two sets never use any values from outside their respective sets and consensus cannot be reached.

(Sufficiency) Besides the method used in \cite{su2017reaching}, we can prove the sufficiency part using the analysis as shown in the proof of Theorem \ref{asyntheorem}, which is for asynchronous updates, since synchronous updates form one special case.
\hfill  $\blacksquare$

We must note that if $\mathcal{G}$ is $(f + 1)$-strictly robust with $l$ hops, then the normal network $\mathcal{G}_{\mathcal{N}}$ is guaranteed to be $(f + 1)$-robust with $l$ hops for any possible cases of the adversary set $\mathcal{A}$ under the $f$-local model. The latter condition is tighter than the former one, but it is not checkable in practice since the identities of the adversary nodes in $\mathcal{A}$ are unknown.
Thus, in Proposition \ref{syn}, we provide the graph condition on $\mathcal{G}$ instead of the condition on the normal network $\mathcal{G}_{\mathcal{N}}$.

We emphasize that our result is a generalization of those in the literature. As mentioned earlier, the work by \cite{su2017reaching} is restricted to the $f$-total model. On the other hand, \cite{dolev1982byzantine} has studied the undirected networks case where the multi-hop communication has unbounded path lengths. In fact, our condition is equivalent to the condition there: (i) $n\geq 3f+1$ and (ii) the graph connectivity is no less than $2f+ 1$.
We can establish the condition (i) by noticing that complete networks have the largest robustness.
By Lemma \ref{robustlemma}\,(2), the robustness of such a graph after removing any $f$ nodes is no greater than $\lceil \frac{n-f}{2}\rceil$. Thus, our result implies $\lceil \frac{n-f}{2}\rceil \geq f+1$, which is equivalent to $n\geq 3f+1$.
For the connectivity condition (ii), note that the graph after removing any $f$ nodes needs to be $(f+1)$-robust with $l$ hops. 
Therefore, a graph satisfying $(f +1)$-strict robustness with $l$ hops has connectivity no less than $2f +1$.

\subsection{Discussions on Different Graph Conditions}\label{secsyn2}

Table~\ref{table1} summarizes graph conditions for resilient consensus under different threat models and update schemes. Notably, there are three conditions under the $f$-total/local model. We clarify the relations and order among them in the next proposition.

\begin{prop}\label{threeconditions}
	For the following graph conditions on any directed graph $\mathcal{G} = (\mathcal{V},\mathcal{E})$ under the $f$-total/local model, where $l\in \mathbb{Z}_+$:
	\vspace*{-3mm}
	
	(\textit{A}) $\mathcal{G}$ is $(2f+1)$-robust with $l$ hops,
	\vspace*{-3mm}
	
	(\textit{B}) $\mathcal{G}$ is $(f+1)$-strictly robust with $l$ hops,
	\vspace*{-3mm}
	
	(\textit{C}) $\mathcal{G}$ is $(f + 1,f+1)$-robust with $l$ hops,
	\vspace*{-3mm}
	
	\noindent  it holds that $(A) \Rightarrow (B)$ and $(B) \Rightarrow (C)$. Moreover, $(C)\nRightarrow (B)$ and $(B)\nRightarrow (A)$.
\end{prop}

\begin{proof}
	($(A) \Rightarrow (B)$) For a graph satisfying $(A)$, take a set $\mathcal{F}$ satisfying the $f$-total/$f$-local model. Select any nonempty disjoint subsets $\mathcal{V}_1, \mathcal{V}_2\subset \mathcal{H}$, where $\mathcal{H}=\mathcal{V}\setminus\mathcal{F}$. Choose node $i\in \mathcal{Z}_{\mathcal{V}_1}^{2f+1}$. Then, after removing nodes in the set $\mathcal{F}$ from $\mathcal{V}$, at most $f$ independent paths are removed. Thus,
	it must hold that $i\in \mathcal{Z}_{\mathcal{V}_1}^{f+1}$ in $\mathcal{G}_{\mathcal{H}}$. Hence, $\mathcal{G}_{\mathcal{H}}$ is $(f+1)$-robust with $l$ hops.
	This is true for any set $\mathcal{F}$. Therefore, $(B)$ holds.

	($(B) \Rightarrow (C)$) We show that $\neg (C) \Rightarrow \neg (B)$. In a graph satisfying $\neg (C)$, for some nonempty disjoint subsets $\mathcal{V}_1, \mathcal{V}_2\subset \mathcal{V}$, at most $f$ nodes in $\mathcal{V}_1, \mathcal{V}_2$ have $f+1$ independent paths originating from the nodes outside. We choose these $f$ nodes as the set $\mathcal{F}$. As a consequence, none of the remaining nodes in $\mathcal{V}_1, \mathcal{V}_2$ has $f+1$ independent paths originating from the nodes outside. Hence this $\mathcal{G}_{\mathcal{H}}$ is not $(f + 1)$-robust with $l$ hops.
	
	($(C)\nRightarrow (B)$, $(B)\nRightarrow (A)$) We show these cases through counter examples in Fig.~\ref{difference}. 
	Suppose that the set $\mathcal{F}$ satisfies 1-local model.
	The graph in Fig.~\ref{difference}(c) is $(2,2)$-robust (satisfying $(C)$), but does not satisfy that any $\mathcal{G}_{\mathcal{H}}$ is $2$-robust where $\mathcal{H}=\mathcal{V}\setminus\mathcal{F}$ (not satisfying $(B)$). The graph in Fig.~\ref{difference}(b) satisfies that any $\mathcal{G}_{\mathcal{H}}$ is $2$-robust (satisfying $(B)$), but this graph is not $3$-robust (not satisfying $(A)$). Moreover, this graph needs one more edge to be $3$-robust as indicated in Fig.~\ref{difference}(a).
\end{proof}

\begin{figure}[t]
	\centering
	\subfigure[]{
		\includegraphics[width=0.9in]{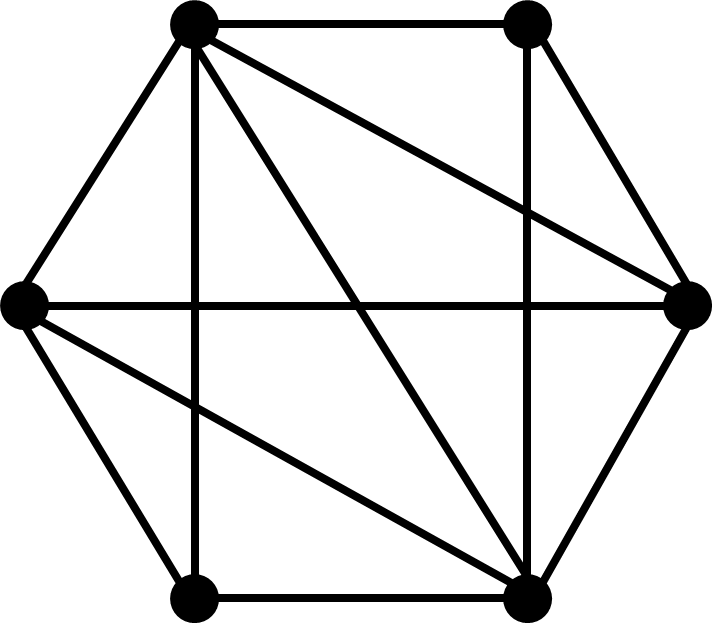}
	}
	\quad
	\subfigure[]{
		\includegraphics[width=0.9in]{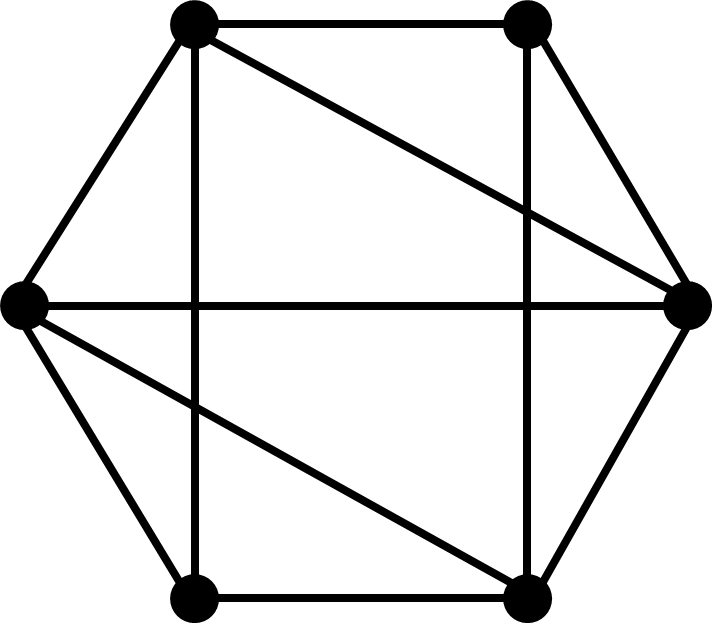}
	}
	\quad
	\subfigure[]{
		\includegraphics[width=0.9in]{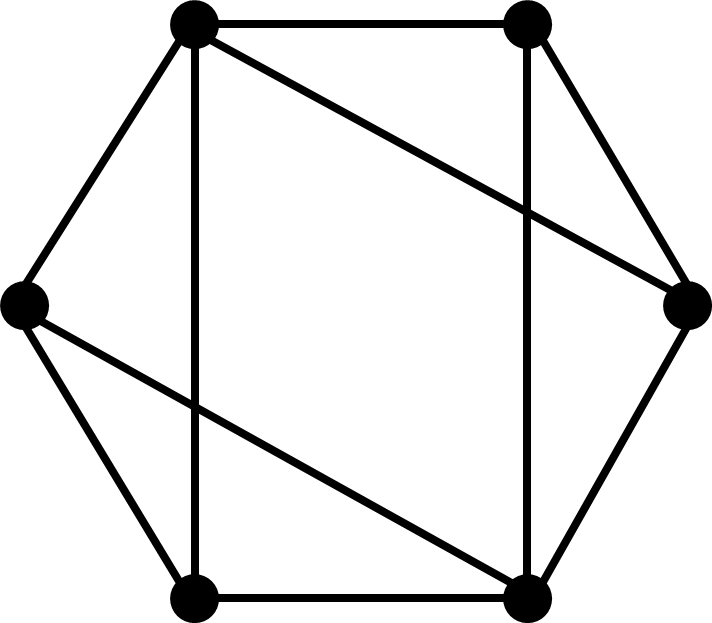}
	}
	\caption{(a) $3$-robust. (b) $2$-strictly robust. (c) $(2,2)$-robust.}
	\label{difference}
\end{figure}

This proposition is of importance since it provides a new characterization for resilient consensus under the $f$-local malicious model. We must first recall that for synchronous update scheme, conditions $(A)$ and $(C)$ are known to be sufficient and necessary conditions, respectively, to achieve resilient consensus under the $f$-local malicious model (\cite{leblanc2013resilient, dibaji2017resilient}). 
On the other hand, we found earlier in this section that condition $(B)$ is a necessary and sufficient condition for synchronous Byzantine consensus under the $f$-local model. It is clear that the Byzantine model includes the case of malicious agents. In view of Proposition~\ref{threeconditions}, we have now established a tighter result as shown in the following corollary.

\begin{cor}\label{synlocalmalicious}
	Consider a directed graph $\mathcal{G} = (\mathcal{V},\mathcal{E})$ with $l$-hop communication, where each normal node updates its value according to the synchronous MW-MSR algorithm with parameter $f$. Under the $f$-local malicious model, resilient asymptotic consensus is achieved with safety interval $\mathcal{S}=\big[ \min x^N[0], \max x^N[0] \big]$ if $\mathcal{G}$ is $(f + 1)$-strictly robust with $l$ hops and only if $\mathcal{G}$ is $(f + 1,f+1)$-robust with $l$ hops.
\end{cor}

Note that condition $(C)$ is a necessary and sufficient condition for the $f$-total malicious model (\cite{leblanc2013resilient, dibaji2017resilient, yuan2021resilient}).

\section{Asynchronous Byzantine Consensus}\label{secasyn}

In practice, normal nodes may not be synchronized nor have access to the current values of all $l$-hop neighbors simultaneously, especially when $l$ is large.
Therefore, in this section, we analyze the asynchronous MW-MSR algorithm under the $f$-local Byzantine model and show our advantages over the conventional ones in terms of threat models and computational complexity.

Our asynchrony setting follows the approach generally assumed in fault-free consensus works (\cite{xiao2006state}; \cite{lin2009consensus}), and those considering the malicious model (e.g., \cite{dibaji2017resilient}). That is, when a normal node updates, it uses the most recently received values of its $l$-hop neighbors. 
Here, we briefly highlight how delays in asynchronous resilient consensus algorithms are handled in computer science area especially through the notion of \textit{rounds} commonly used in, e.g., \cite{azadmanesh2002asynchronous,abraham2004optimal,sakavalas2020asynchronous}.
There, each node labels its updated value with round $r$, representing the number of transmissions made so far. Moreover, if a normal node wants to update its next value with round $r+1$, it has to wait until receiving a sufficient number of values labeled with the same round $r$. This may cause potentially large delays in making the $(r + 1)$th update for some nodes. 
We note that the use of rounds can create further 
problems due to following the fixed order in the indices of rounds. 
That is, node $i$ may receive the value of round $r+1$
before the one of round $r$ from its neighbor. This may occur even along a non-faulty path.
In this case, the old data from round $r$ will be
used even though more recent data of round $r+1$ is available at the node. 
This is because the FIFO (first-in-first-out) message receiving mechanism is applied in \cite{abraham2004optimal}, \cite{sakavalas2020asynchronous}. 
However, in our asynchrony setting, these issues do not arise, and node $i$ will use the most recently received values of all neighbors whenever node $i$ chooses to update.

\subsection{Consensus Analysis}

When communication among nodes is subject to possible
time delays, we can write the control input as
\begin{equation}
	u_i[k]=\sum_{j\in \mathcal{N}_i^{l-}} a_{ij}[k]x_j^P[k-\tau_{ij}^P[k]],  
\end{equation}
where $x_j^P[k]$ denotes the value of node $j$ at time $k$ sent along path $P$, $a_{ij}[k]$ is the time-varying weight, and $\tau_{ij}^P[k]\in \mathbb{Z}_+$ denotes the delay in this $(j,i)$-path $P$ at time $k$.
The delays are time varying and may be different in each path, but we assume the common upper bound $\tau$ in any normal path $P$ (i.e., all nodes on path $P$ are normal) as
\begin{equation}
	0\leq \tau_{ij}^P[k] \leq \tau,\medspace j\in \mathcal{N}_i^{l-}, \medspace k\in \mathbb{Z}_+.
\end{equation}
Hence, each node $i \in \mathcal{N}$ becomes aware of the value
of each of its normal $l$-hop neighbor $j$ in each normal $(j,i)$-path $P$ at least once in $\tau$ time steps, but
possibly at different time instants (\cite{dibaji2017resilient}). 
Note that the delay bound need not be known by normal nodes.
Finally, we outline the asynchronous MW-MSR algorithm as follows.
\begin{enumerate}
	\item At $k\geq 0$, each node $i \in \mathcal{N}$ independently chooses to update or not. 
	\item If it chooses not to update, then $x_i[k+1]=x_i[k]$ and it does not transmit its own message.
	
	\item Otherwise, it will use the most recently received values of $\forall j\in \mathcal{N}_i^{l-}$ on each $l$-hop path to update its value following steps 2 and 3 in Algorithm~1. Then it transmits its new message to $\forall j\in \mathcal{N}_i^{l+}$.
\end{enumerate}

If node $i$ does not receive any value along some path $P$ originating from $j\in \mathcal{N}_i^{l-}$ (i.e., the crash model), then it considers this value on path $P$ as one empty value and discards this value when it applies Algorithm~1.

To proceed with our analysis, we introduce some notations. Let $D[k]$ be a diagonal matrix whose $i$th entry is
given by $d_i[k]=\sum_{j=1}^{n} a_{ij}[k].$ Then,
let the matrices $A_\gamma[k]\in \mathbb{R}^{n\times n}$ for $ 0\leq \gamma \leq \tau$, and $L_{\tau}[k]\in \mathbb{R}^{n\times (\tau +1)n}$ be
\begin{equation}   
	A_\gamma[k]=
	\begin{cases}
		a_{ij}[k]  &  \text{if $i\neq j $ and $\tau_{ij}[k]=\gamma$,} \\
		0 &  \text{otherwise,}
	\end{cases}               
\end{equation}
and $L_{\tau}[k]=\Big[ D[k]-A_0[k] \medspace \medspace\medspace\medspace -A_1[k] \medspace \medspace\medspace\medspace \cdots \medspace\medspace \medspace\medspace -A_{\tau}[k] \Big].$
Now, the control input can be expressed as
\begin{equation}
	\begin{array}{lll} 
		u^N[k] =-L_{\tau}^N[k]z[k],\\
		u^A[k] : \textup{arbitrary,}
	\end{array}
\end{equation}
where $z[k]= [x[k]^T x[k-1]^T \cdots  x[k-\tau]^T]^T$ is a $(\tau+1)n$-dimensional vector for $k\geq0$ and $L_{\tau}^N[k]$ is a matrix formed by the first $n_N$ rows of $L_{\tau}[k]$.
Here, $z[0] = [x[0]^T 0^T \cdots 0^T]^T$.
Then, the agent dynamics can be written as
\begin{equation} \label{system2}
	x[k+1]=\Gamma[k] z[k] +   \begin{bmatrix} 0  \\ I_{n_A} \end{bmatrix}u^A[k],
\end{equation}
where $\Gamma[k]$ is an  $n\times(\tau+1)n$ matrix given by $\Gamma[k] = \begin{bmatrix} I_n &  0 \end{bmatrix} -   \begin{bmatrix} L_{\tau}^N[k]^T  & 0 \end{bmatrix}^T. $
The safety interval is given by 
\begin{equation}  \label{safety2}
	\mathcal{S}_{\tau}=\Big[ \min z^N[0], \max z^N[0] \Big].            	
\end{equation}
The following is the main result of this paper.
It provides a necessary and sufficient condition for the asynchronous MW-MSR algorithm achieving Byzantine consensus.

\begin{thm}\label{asyntheorem}
	Consider a directed graph $\mathcal{G} = (\mathcal{V},\mathcal{E})$ with $l$-hop communication, where each normal node updates its value according to the asynchronous MW-MSR algorithm with parameter
	$f$. Under the $f$-local Byzantine model for the adversarial nodes, resilient asymptotic
	consensus is achieved with the safety interval given by \eqref{safety2}
	if and only if $\mathcal{G}$ is $(f + 1)$-strictly robust with $l$ hops.
\end{thm}


\begin{proof}
	The necessity part follows from Proposition \ref{syn}. For sufficiency, we first show that the safety condition holds. For $k = 0$, by \eqref{safety2}, we have $x_i[0] \in \mathcal{S}_{\tau} ,\forall i \in \mathcal{N}$. For $k = 1$, $\forall i \in \mathcal{N}$, the right-hand side of \eqref{system2} becomes convex combinations of values in the interval  $[\min z^N[0], \max z^N[0]]=\mathcal{S}_{\tau}$. Thus, $x_i[1] \in \mathcal{S}_{\tau} ,\forall i \in \mathcal{N}$.
	Next, for $k \geq 1$, define two variables by
	\begin{equation}
		\begin{array}{lll} 
			\overline{x}_\tau[k] =\max \left( x^N[k], x^N[k-1],\dots, x^N[k-\tau]\right),\\
			\underline{x}_\tau[k] = \min \left( x^N[k], x^N[k-1],\dots, x^N[k-\tau]\right).
		\end{array}
	\end{equation}
	For $k \geq 2$, from step 2 of Algorithm~1, we obtain $x_i[k+1] \leq \max \left( x^N[k], x^N[k-1],\dots, x^N[k-\tau]\right), \forall i \in \mathcal{N}$. Also, for $\tau'=1,2,\dots,\tau$, it holds that
	\begin{equation*}
		x_i[k+1-\tau']\leq \max \left( x^N[k], x^N[k-1],\dots, x^N[k-\tau]\right),
	\end{equation*}
	$\forall i \in \mathcal{N}$. Hence, $\overline{x}_\tau[k]$ is nonincreasing in time as
	\begin{equation*}
		\begin{aligned}
			\overline{x}_\tau&[k+1] =\max \left( x^N[k+1], x^N[k],\dots, x^N[k+1-\tau]\right)\\
			&\leq \max \left( x^N[k], x^N[k-1],\dots, x^N[k-\tau]\right)=\overline{x}_\tau[k].
		\end{aligned}
	\end{equation*}
	We can similarly prove that $\underline{x}_\tau[k]$ is nondecreasing in time. Thus, we have shown the safety condition.
	
	From above, $\overline{x}_\tau[k]$ and $\underline{x}_\tau[k]$ are monotone and bounded, and thus both of their limits exist and are denoted by $\overline{x}_\tau^*$ and $\underline{x}_\tau^*$, respectively.
	We prove by contradiction that $\overline{x}_\tau^*=\underline{x}_\tau^*$. Assume that $\overline{x}_\tau^*>\underline{x}_\tau^*$
	and $\alpha$ lower bounds the nonzero entries of $\Gamma[k]$. Choose $\epsilon_0 > 0$ small enough that $\overline{x}_\tau^*-\epsilon_0 > \underline{x}_\tau^*+\epsilon_0$. Fix 
	\begin{equation}\label{ep}
		\epsilon<\frac{\epsilon_0\alpha^{(\tau+1)n_N}}{(1-\alpha^{(\tau+1)n_N})}, \medspace 0<\epsilon<\epsilon_0.
	\end{equation}
	Define the sequence $\{\epsilon_\gamma\}$ by
	$\epsilon_{\gamma+1}= \alpha\epsilon_\gamma-(1-\alpha)\epsilon, \medspace \gamma=0,1,\dots,(\tau+1)n_N-1.$
	So we have $0 < \epsilon_{\gamma+1} < \epsilon_{\gamma}$ for all $\gamma$. In particular, they are positive because by \eqref{ep},
	\begin{equation*}
		\begin{aligned}
			\epsilon_{(\tau+1)n_N}&= \alpha^{(\tau+1)n_N}\epsilon_0- \sum_{m=0}^{(\tau+1)n_N-1}\alpha^m(1-\alpha)\epsilon\\
			&= \alpha^{(\tau+1)n_N}\epsilon_0-(1-\alpha^{(\tau+1)n_N})\epsilon>0.
		\end{aligned}
	\end{equation*}
	Take $k_\epsilon \in \mathbb{Z}_+$ such that $\overline{x}_\tau[k]<\overline{x}_\tau^*+\epsilon$ and $\underline{x}_\tau[k]>\underline{x}_\tau^*-\epsilon$ for $k\geq k_\epsilon$. Such $k_\epsilon$ exists due to the convergence of $\overline{x}_\tau[k]$ and $\underline{x}_\tau[k]$. Then we can define the two disjoint sets as 
	\begin{equation*}
		\begin{aligned}
			\mathcal{Z}_{1\tau}(k_\epsilon+\gamma,\epsilon_\gamma)&=\{j\in \mathcal{N}: x_j[k_\epsilon+\gamma]>\overline{x}_\tau^*-\epsilon_\gamma\}, \\
			\mathcal{Z}_{2\tau}(k_\epsilon+\gamma,\epsilon_\gamma)&=\{j\in \mathcal{N}: x_j[k_\epsilon+\gamma]<\underline{x}_\tau^*+\epsilon_\gamma\}.
		\end{aligned}
	\end{equation*}
	We show that one of them becomes empty in finite steps, which contradicts the assumption on $\overline{x}_\tau^*$ and $\underline{x}_\tau^*$ being the limits. Consider $\mathcal{Z}_{1\tau}(k_\epsilon,\epsilon_0)$. Due to the definition of $\overline{x}_\tau[k]$ and its limit $\overline{x}_\tau^*$, one or more normal nodes are in the union of the sets $\mathcal{Z}_{1\tau}(k_\epsilon+\gamma,\epsilon_\gamma)$ for $0 \leq\gamma\leq\tau + 1$. We claim that $\mathcal{Z}_{1\tau}(k_\epsilon,\epsilon_0)$ is in fact nonempty. To prove this, it is sufficient to show that if a normal node $j$ is not in $\mathcal{Z}_{1\tau}(k_\epsilon+\gamma,\epsilon_\gamma)$, then it is not in $\mathcal{Z}_{1\tau}(k_\epsilon+\gamma+1,\epsilon_{\gamma+1})$ for $\gamma=0,\dots,\tau$.
	Suppose that node $j$ satisfies $x_j[k_\epsilon+\gamma]\leq \overline{x}_\tau^*-\epsilon_\gamma$. The values greater than $\overline{x}_\tau[k_\epsilon+\gamma]$ are ignored in step 2 of Algorithm~1. Thus, its next value is bounded as
	\begin{equation}\label{converge}
		\begin{aligned}
			x_j&[k_\epsilon+\gamma+1] \leq (1-\alpha)\overline{x}_\tau[k_\epsilon+\gamma]+\alpha(\overline{x}_\tau^*-\epsilon_\gamma)\\
			&\leq (1-\alpha)(\overline{x}_\tau^*+\epsilon)+\alpha(\overline{x}_\tau^*-\epsilon_\gamma)\\
			&\leq \overline{x}_\tau^*-\alpha\epsilon_\gamma+(1-\alpha)\epsilon
			=\overline{x}_\tau^*-\epsilon_{\gamma+1}.
		\end{aligned}
	\end{equation}
	Thus, node $j$ is not in $\mathcal{Z}_{1\tau}(k_\epsilon+\gamma+1,\epsilon_{\gamma+1})$. Then, $|\mathcal{Z}_{1\tau}(k_\epsilon+\gamma,\epsilon_\gamma)|$ is nonincreasing for $\gamma=0,\dots,\tau+1$. Similarly, $\mathcal{Z}_{2\tau}(k_\epsilon,\epsilon_0)$ is nonempty too.

	Since $\mathcal{G}$ is $(f + 1)$-strictly robust with $l$ hops under the $f$-local model, $\mathcal{G}_{\mathcal{N}}$ must be $(f + 1)$-robust with $l$ hops w.r.t. $\mathcal{A}$.
	Thus, $\exists i \in \mathcal{V}_a$, such that $ i \in \mathcal{Y}_{\mathcal{V}_a}^{f+1}$ in Definition~\ref{rs-robust}, where $\mathcal{V}_a $ is one of the nonempty disjoint sets $\mathcal{Z}_{1\tau}(k_\epsilon,\epsilon_0)$ and $\mathcal{Z}_{2\tau}(k_\epsilon,\epsilon_0)$.
	Suppose that $ i \in \mathcal{Y}_{\mathcal{V}_a}^{f+1}$ and $\mathcal{V}_a =  \mathcal{Z}_{1\tau}(k_\epsilon,\epsilon_0)$.
	By the argument above, node $i$'s normal neighbors outside $\mathcal{Z}_{1\tau}(k_\epsilon,\epsilon_0)$ will not be in $\mathcal{Z}_{1\tau}(k_\epsilon+\gamma,\epsilon_\gamma)$ for $0\leq\gamma\leq\tau$. By step 2 of Algorithm~1, one value of these neighbors upper bounded by $\overline{x}_\tau^*-\epsilon_\tau$ will be used in the updates of node $i$ at any time (e.g., at time $k_\epsilon+\tau$) since node $i$ can only remove the smallest values of which the cardinality of the MMC is $f$. Thus,
	\begin{equation*}
		x_i[k_\epsilon+\tau+1]\leq (1-\alpha)\overline{x}_\tau[k_\epsilon+\tau]+\alpha(\overline{x}_\tau^*-\epsilon_\tau).
	\end{equation*}
	By \eqref{converge}, we have $x_i[k_\epsilon+\tau+1] \leq \overline{x}_\tau^*-\epsilon_{\tau+1}$. If $ \mathcal{V}_a =  \mathcal{Z}_{1\tau}(k_\epsilon,\epsilon_0)$, then node $i$ goes outside of $\mathcal{Z}_{1\tau}(k_\epsilon+\tau+1,\epsilon_{\tau+1})$ after $\tau + 1$ steps. Consequently, $\left|\mathcal{Z}_{1\tau}(k_\epsilon+\tau+1,\epsilon_{\tau+1}) \right| < \left| \mathcal{Z}_{1\tau}(k_\epsilon,\epsilon_0)\right| $. Likewise, if $\mathcal{V}_a =  \mathcal{Z}_{2\tau}(k_\epsilon,\epsilon_0)$, then $\left|\mathcal{Z}_{2\tau}(k_\epsilon+\tau+1,\epsilon_{\tau+1}) \right| < \left| \mathcal{Z}_{2\tau}(k_\epsilon,\epsilon_0)\right| $.
	Since $|\mathcal{N}|=n_N$, we can repeat the steps above until one of $\mathcal{Z}_{1\tau}(k_\epsilon+\tau+1,\epsilon_{\tau+1})$ and $\mathcal{Z}_{2\tau}(k_\epsilon+\tau+1,\epsilon_{\tau+1})$ remains empty indefinitely, and it takes no more than $(\tau + 1)n_N$ steps. This contradicts the assumption that $\overline{x}_\tau^*$ and $\underline{x}_\tau^*$ are the limits. Therefore, we obtain $\overline{x}_\tau^*=\underline{x}_\tau^*$.
\end{proof}

\subsection{Comparison with Conventional Methods}

In this part, we outline our advantages over the conventional works. They are highlighted in the following four aspects: 
(i) Our algorithm does not use ``rounds'' that can cause possibly large delays in consensus forming;
(ii) we consider the $f$-local model;
(iii) our graph condition is tight and generalizes the ones in the literature for both synchronous and asynchronous cases;
(iv) the algorithm is computationally more efficient.

\subsubsection{Advantages in Threat Models and Graph Conditions}\label{sec_advantages}

In what follows, we discuss further details about these advantages.
Specifically, the $f$-total model in \cite{su2017reaching}, \cite{sakavalas2020asynchronous} can be viewed as a special case of the $f$-local model, and thus the condition stated in Theorem \ref{asyntheorem} is also sufficient for the $f$-total Byzantine model.
We emphasize that the $f$-local model is more suitable for a large scale network because it locally focuses on each node with a small $f$-total model. If the locations of adversary nodes are spread in a more uniform way over the network, then the total tolerable number of adversaries can be very large. However, with the same number of adversary nodes, the $f$-total model requires much more connections in the network. See the example in Fig.~\ref{1lcoal}(b) and the simulation in Section~\ref{sec_sim}.

Observe that the condition in Theorem \ref{asyntheorem} is the same for the synchronous case in Section~\ref{secsyn} and \cite{su2017reaching}, which indicates that it makes the system sufficiently resilient to the influence of asynchrony and communication delays. 
It also appeared in \cite{tseng2015fault}, \cite{sakavalas2020asynchronous} for synchronous and asynchronous schemes, respectively, which study the special case of unbounded path length $l\geq l^*$, where $l^*$ is the length of the longest cycle-free path in the network.

Similar to the discussion in Section~\ref{secsyn2}, based on Theorem~\ref{asyntheorem}, we can obtain a tight result for the resilient consensus under the $f$-total/local malicious model for the asynchronous update scheme. For this case, it is known that $(2f +1)$-robustness with $l$ hops is a sufficient condition (see Table~\ref{table1}) while $(f +1,f+1)$-robustness with $l$ hops is a necessary condition; see, e.g., \cite{leblanc2013resilient}, \cite{dibaji2017resilient} for the one-hop case and \cite{yuan2021resilient} for the multi-hop case. Therefore, the following corollary gives a tighter graph condition for asynchronous resilient consensus under the malicious model in view of Proposition~\ref{threeconditions} and Theorem~\ref{asyntheorem}.

\begin{cor}\label{asynlocalmalicious}
	Consider a directed graph $\mathcal{G} = (\mathcal{V},\mathcal{E})$ with $l$-hop communication, where each normal node updates its value according to the asynchronous MW-MSR algorithm with parameter $f$. Under the $f$-local/total malicious model, resilient asymptotic consensus is achieved with safety interval \eqref{safety2} if $\mathcal{G}$ is $(f + 1)$-strictly robust with $l$ hops and only if $\mathcal{G}$ is $(f + 1,f+1)$-robust with $l$ hops (under the corresponding models of f-local/total).
\end{cor}

\subsubsection{Advantages in Computational Complexity}
Finally, we highlight that our MW-MSR algorithm is more light weighted and efficient in terms of computational complexity in comparison with the flooding-based algorithm in \cite{sakavalas2020asynchronous}.

To show this, we first outline the structure of the algorithm in \cite{sakavalas2020asynchronous}.
Intuitively, the algorithm there can be divided into two parts: Verification of the received values and the MSR algorithm (called Filter and Average algorithm).
More specifically, each node is required to send its value to the entire network at the beginning of each asynchronous round. Then in the verification part, for each possible set of Byzantine nodes $\mathcal{F}$ (satisfying the $f$-total model), each normal node $i$ receives values from the neighbors and for each received value, it verifies if this value is consistent in the paths excluding the nodes in set $\mathcal{F}$. Then node $i$ has to wait for enough verified values with round $r$ as the input for the Filter-and-Average part to obtain its new value.

The Filter-and-Average algorithm and our MW-MSR algorithm are similar, but the main difference is that 
the former algorithm uses verified values with round $r$ as inputs and the MW-MSR algorithm uses the most recent values of $l$-hop neighbors on each $l$-hop path. Hence, all the operations before the Filter-and-Average algorithm in the main algorithm for verification in \cite{sakavalas2020asynchronous} are additional in terms of computation. 
Besides, the verification algorithm there should be executed for each possible set $\mathcal{F}$, i.e., at least $\binom{n}{f}$ executions of the main algorithm on each node for each asynchronous round. 
Although this can be
executed in parallel threads (one $\mathcal{F}$ per thread), it still requires a huge amount of computation resources and memory to verify and store the values from the nodes in the entire network.
Even for the case of $l\geq l^*$, the computational complexity of the MW-MSR algorithm is less than the algorithm in \cite{sakavalas2020asynchronous}.

Why the verification part is essential in \cite{sakavalas2020asynchronous} is partially because of their asynchrony setting based on rounds and the verification can prevent the duplication of messages of normal nodes with the same round $r$. 
In contrast, in our asynchrony setting, we need not check the correctness of the received values and we simply use the most recent value for each $l$-hop path (hence, no duplication). Thus, we can fully utilize the ability of MW-MSR algorithm to filter the extreme values that could possibly be manipulated by Byzantine nodes. 
The trade-off is that we can only guarantee $\Delta x_{\tau}[k]=\max z^N[k]- \min z^N[k]$ to be nonincreasing, while for the round based asynchrony, $\Delta x[r]=\max x^N[r]- \min x^N[r]$ is guaranteed to be nonincreasing.
Besides, since our algorithm is iterative and only requires values and topology information up to $l$ hops away, our algorithm is more distributed compared to that in \cite{sakavalas2020asynchronous}.

\section{Numerical Examples}\label{sec_sim}

In this section, we carry out simulations to illustrate the efficacy of the proposed MW-MSR algorithms. Through the example, we also demonstrate our advantages in tolerating more Byzantine agents under the same network setting compared to the flooding-based algorithm.

\subsection{Simulation in a Small Network: Larger Relay Range Improves Strict Robustness}\label{sec_small}

Consider the 7-node network in Fig.~\ref{1lcoal}(a). This graph is not $2$-strictly robust with one hop, but is $2$-strictly robust with 2 hops. Suppose that node 7 is Byzantine and is capable to send six different values to its six neighbors.
Let the initial normal states be $x^N[0]=[1\ 2\ 4\ 9\ 8\ 9]^T$. 
We start with the synchronous case. For this case, according to \cite{vaidya2012iterative}, \cite{leblanc2013resilient}, the current graph does not meet the condition for $1$-total Byzantine model. As shown in Fig.~\ref{one-hop-state}, consensus among normal nodes cannot be reached in this network with one-hop communication.
Here, the Byzantine node 7 transmits six different values indicated by red dashed lines in Fig.~\ref{one-hop-state}(a).

\begin{figure}[t]
	\centering
		\includegraphics[width=3.2in,height=1.5in]{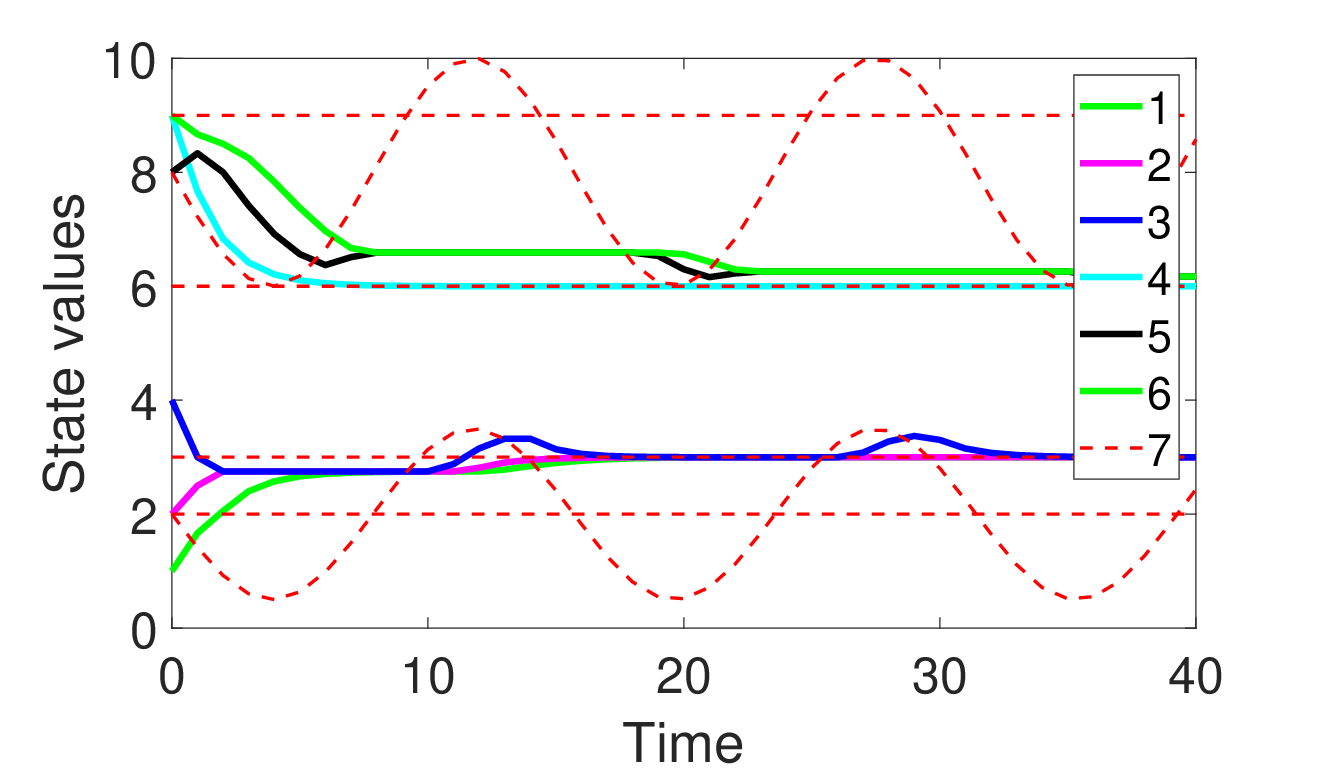}
	\vspace{-8pt}
	\caption{Time responses of the synchronous one-hop W-MSR algorithm in the 7-node network of Fig.~\ref{1lcoal}(a).}
	\label{one-hop-state}
\end{figure}

Then, we examine the synchronous two-hop MW-MSR algorithm in this network. Suppose that node 7 manipulates all the values (including its own value and the relayed values) sent to node 1 as the same value sent to node 1 in the one-hop case. For the other nodes receiving values from node 7, the situations are similar. Observe in Fig.~\ref{two-hop-7}(a) that Byzantine consensus is indeed achieved with two-hop communication.

Next, we perform simulations for the asynchronous two-hop MW-MSR algorithm under the same attack. Let the normal nodes update in an asynchronous periodic sense, which means that for nodes 1, 2, 3, 4, 5, and 6, they update in every 1, 2, 5, 6, 4, 3 steps, respectively (all nodes update once at $k=0$). The time delays for the values from one-hop neighbors and two-hop neighbors are set as 0 and 1 step, respectively. Thus, in the current setting, we can choose $\tau=7$.
The results of the asynchronous two-hop algorithm are presented in Fig.~\ref{two-hop-7}(b).
Observe that Byzantine consensus is achieved although delays have some effects and the convergence takes more time than the synchronous algorithm. We can also notice that the consensus error for $z[k]$, i.e., $\Delta x_{\tau}[k]=\max z^N[k]- \min z^N[k]$ is nonincreasing while $\Delta x_{0}[k]$ is not. This observation also verifies the theoretical results in Theorem~\ref{asyntheorem}. We finally note that the flooding algorithm in \cite{sakavalas2020asynchronous} can achieve asynchronous Byzantine consensus in this network. However, it is achieved with 6-hop communication; this is the length of the longest cycle-free path in this network, required for the flooding-based approach.

%
%
%

\begin{figure}[t]
	\centering
	\subfigure[\scriptsize{Synchronous updates.}]{
		\includegraphics[width=3.2in,height=1.5in]{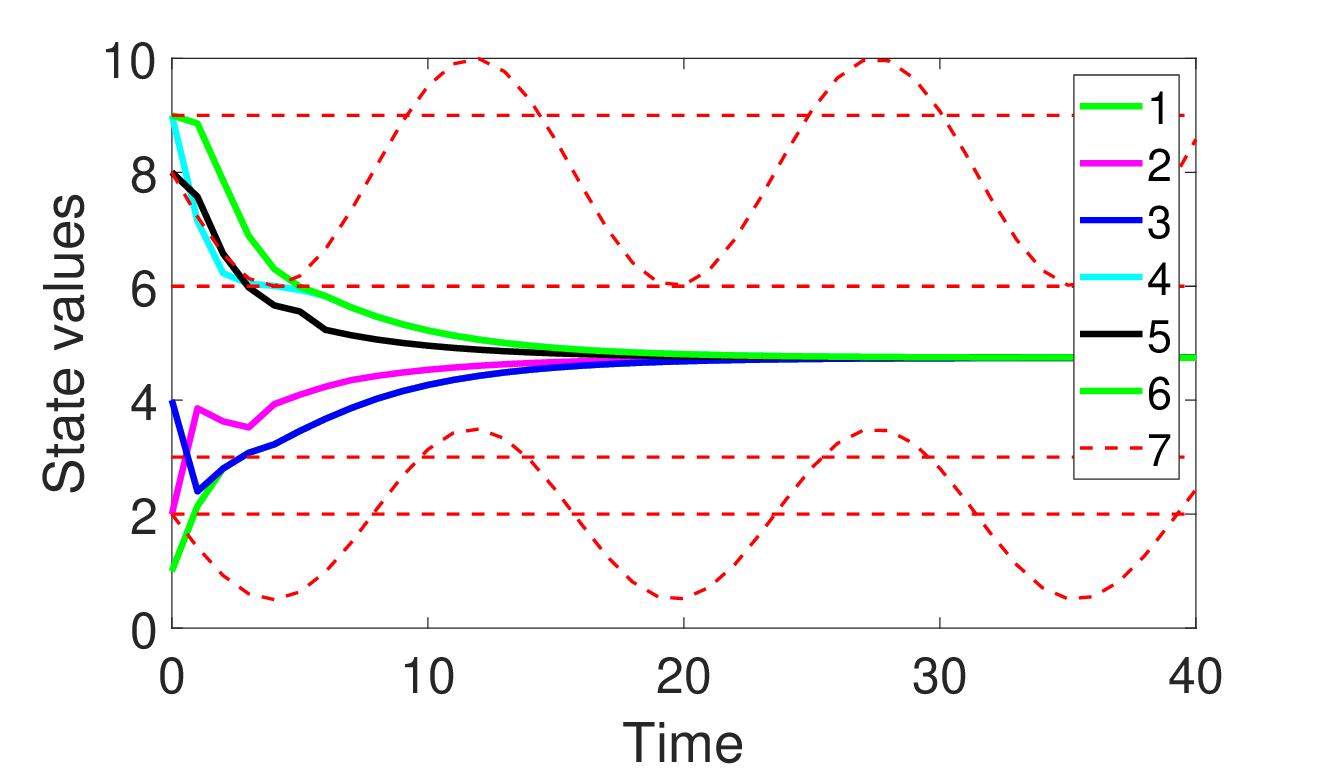}
	}
	
	\vspace{-6pt}
	\subfigure[\scriptsize{Asynchronous updates with delays.}]{
		\includegraphics[width=3.2in,height=1.5in]{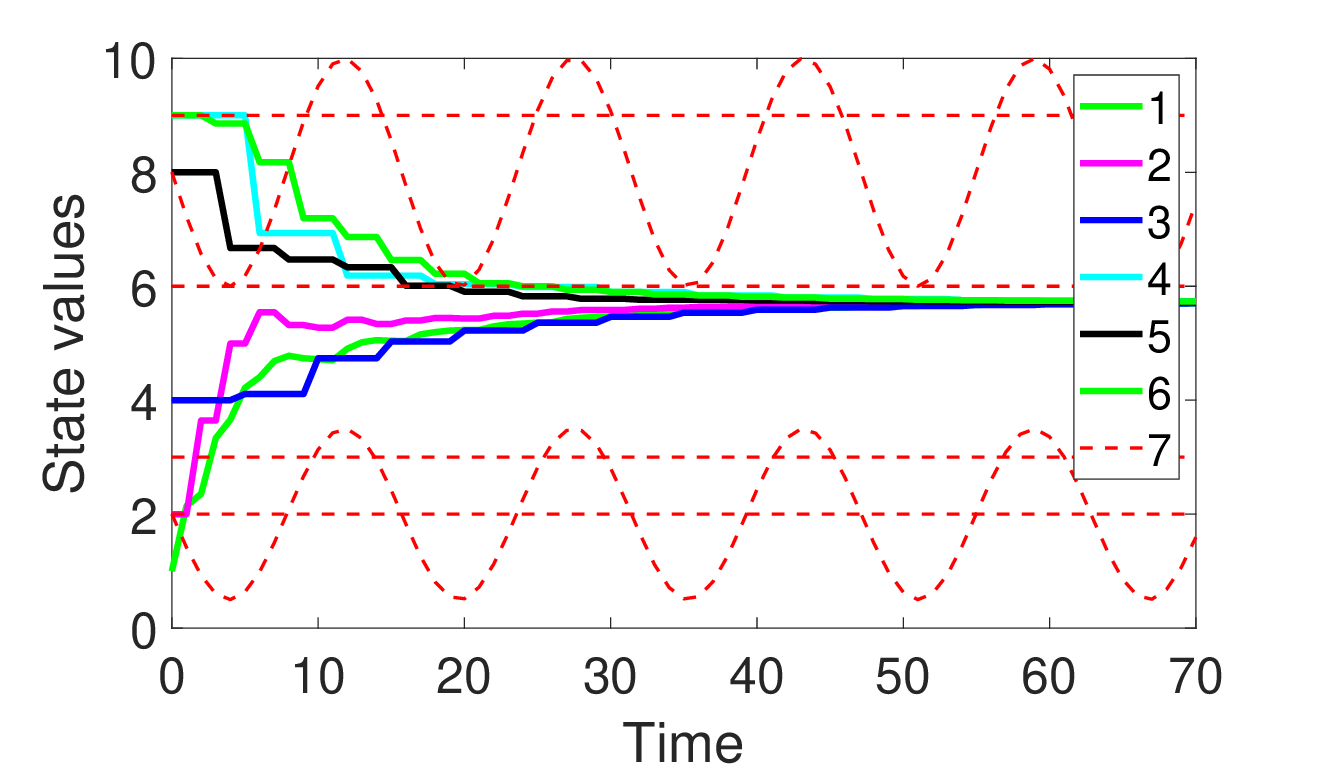}
	}
	\vspace{-8pt}
	\caption{Time responses of the two-hop MW-MSR algorithm in the 7-node network of Fig.~\ref{1lcoal}(a).}
	\label{two-hop-7}
\end{figure}

\subsection{Simulation in a Medium-sized Network: $f$-local versus $f$-total } \label{sec17node}

In this part, we perform further comparisons and show that our algorithm for the $f$-local model can tolerate more Byzantine agents than the flooding-based algorithm for the $f$-total model from \cite{sakavalas2020asynchronous}.
For this purpose, we apply our algorithm in the 17-node network in Fig.~\ref{1lcoal}(b).
As mentioned in Section~\ref{sec_strict_robustness}, this graph is not $2$-strictly robust with one hop, but is $2$-strictly robust with 2 hops under the 1-local model.

Assume that nodes 1 and 15 are Byzantine. Node 15 transmits four distinct values to its neighbors while node 1 maintains a constant value (indicated in red dashed lines in Fig.~\ref{two-hop-17}). Let the initial states of the normal agents fall within the range of $(0,40)$.
According to the results in \cite{vaidya2012iterative,leblanc2013resilient}, this graph fails to satisfy the criteria for either the $1$-local or the $1$-total Byzantine model even for synchronous updates. Consequently, in Fig.~\ref{two-hop-17}(a), Byzantine consensus is not achieved by the one-hop MW-MSR algorithm, which is equivalent to the algorithms in \cite{vaidya2012iterative,leblanc2013resilient}.

\begin{figure}[t]
	\centering
	\includegraphics[width=3.2in,height=1.5in]{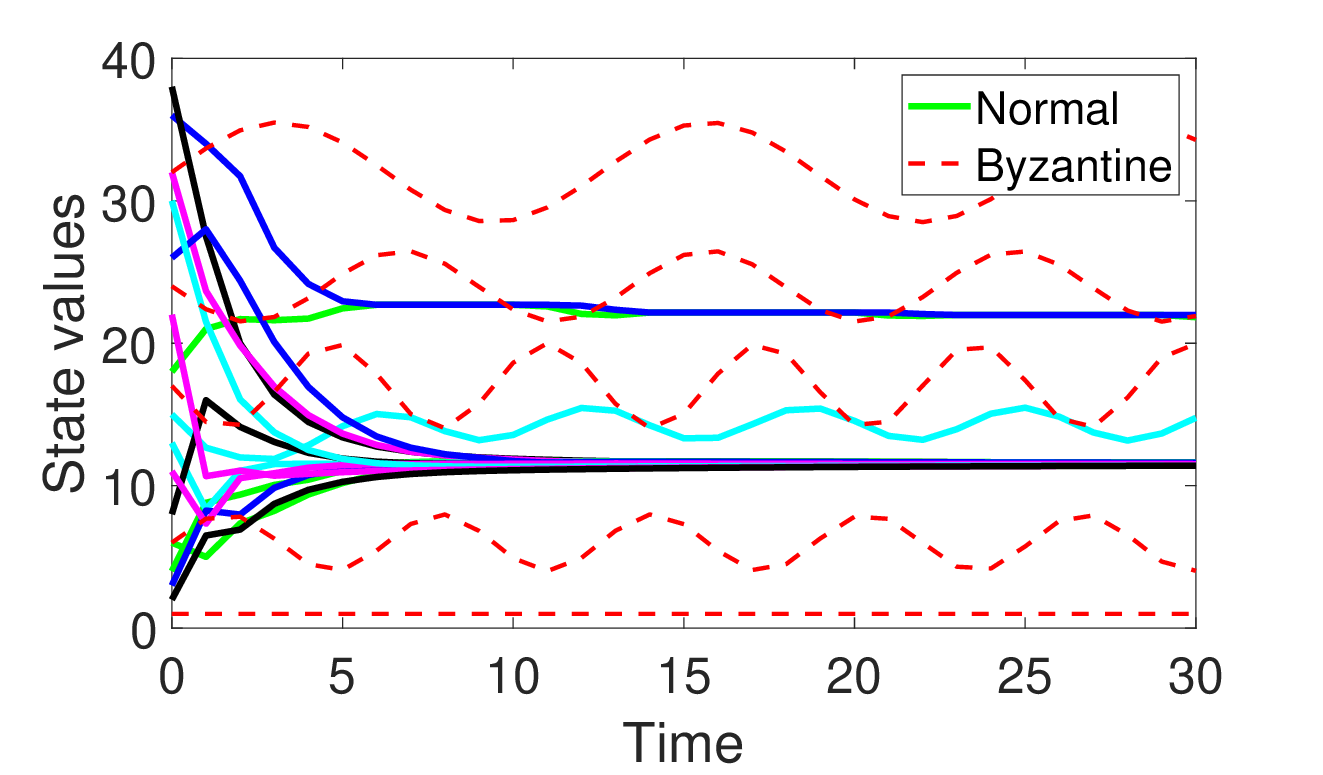}
	\vspace{-8pt}
	\caption{Time responses of the synchronous one-hop W-MSR algorithm in the 17-node network of Fig.~\ref{1lcoal}(b).}
	\label{one-hop-17}
\end{figure}

Then, we perform simulations for the synchronous and asynchronous two-hop MW-MSR algorithm under the same attacks, respectively. 
The results for the synchronous algorithm are given in Fig. \ref{two-hop-17}(b) and Byzantine consensus is achieved.
Next, let the normal nodes update asynchronously with delays in communication. 
Observe that Byzantine consensus is also achieved as shown in Fig. \ref{two-hop-17}(c), although the final stage of consensus takes longer due to the communication delays. These simulations clearly verifies the effectiveness of the proposed algorithm.

As a comparison, the flooding-based algorithm (\cite{sakavalas2020asynchronous}) for the $f$-total model cannot solve the Byzantine consensus under the same attack scenario. The reason is that for their algorithm to tolerate two Byzantine agents, the minimum in-degree of the graph needs to be at least $2f+1=5$, which is apparently not satisfied in the 17-node network.
Actually, our algorithm can achieve Byzantine consensus even in larger networks with more Byzantine nodes.
As discussed in the beginning of Section~\ref{sec_advantages}, if the locations of Byzantine nodes are well spread in the network, then the total tolerable number of Byzantine nodes can be very large.
This is because the erroneous influences from Byzantine nodes are also bounded by the relay range.
However, this situation clearly exceeds the capability of the flooding-based algorithm in \cite{sakavalas2020asynchronous}, where a Byzantine node can have erroneous influences on all the nodes in the network.

\begin{figure}[t]
	\centering
	\subfigure[\scriptsize{Synchronous updates.}]{
		\includegraphics[width=3.2in,height=1.5in]{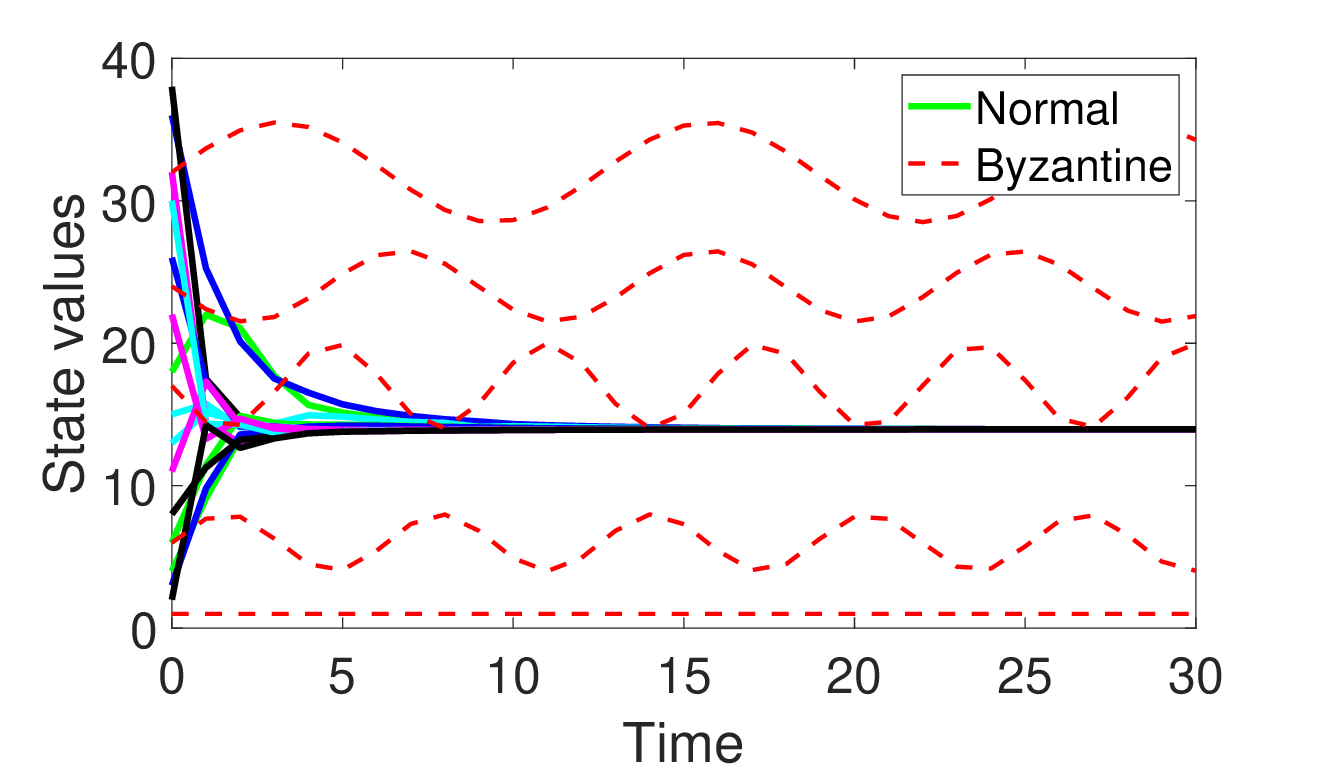}
	}
	
	\vspace{-6pt}
	\subfigure[\scriptsize{Asynchronous updates with delays.}]{
		\includegraphics[width=3.2in,height=1.5in]{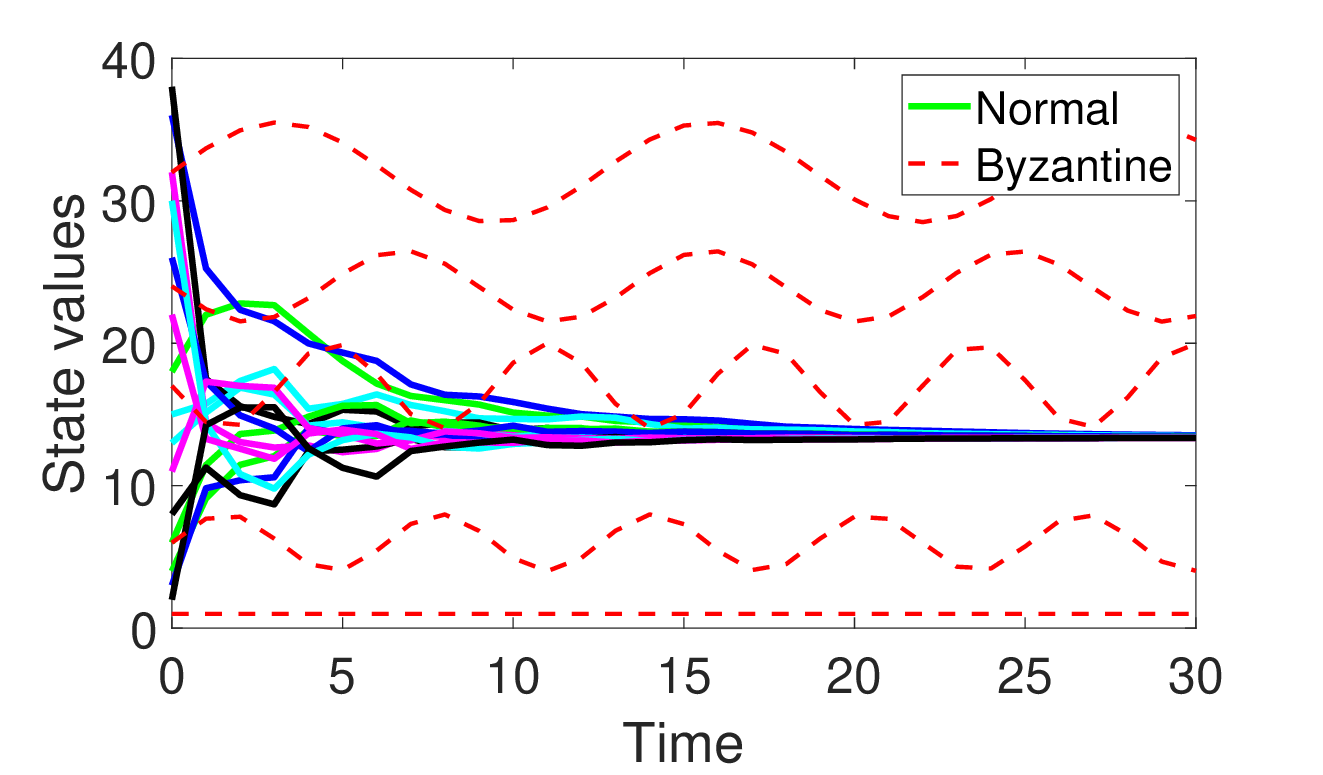}
	}
	\vspace{-8pt}
	\caption{Time responses of the two-hop MW-MSR algorithm in the 17-node network of Fig.~\ref{1lcoal}(b).}
	\label{two-hop-17}
\end{figure}

\section{Conclusion}
We have solved the approximate Byzantine consensus problem under asynchronous updates with time delays in the agents' communication. Our approach is based on the multi-hop weighted MSR algorithm. We have specifically provided a tight necessary and sufficient graph condition for the network using the MW-MSR algorithm for Byzantine consensus.
It is expressed using the notion of $r$-strictly robust graphs with $l$ hops. An important implication of our results is that under the $f$-total/local Byzantine model, the graph condition remains the same even if the algorithm becomes asynchronous and the communication is subject to time delays. 
Our analysis has led us to tighter robust graph conditions for the case of the malicious model than those known in the literature as well. 
Moreover, our algorithm is iterative and requires only local information and topology for each node, and hence it is more light-weighted and distributed compared to the conventional flooding-based algorithms.

\fontsize{8pt}{8.5pt}\selectfont

\end{document}